\documentclass[twocolumn,notitlepage,superscriptaddress, nofootinbib,nobibnotes, aps,prd,10pt]{revtex4-2}

\usepackage[english]{babel}
\usepackage{graphicx}
\graphicspath{ {./Figures/} }
\usepackage{tikz}
\usetikzlibrary{shapes.misc, positioning, arrows}
\usepackage{verbatim}
\usepackage{bbold}
\usepackage{bbm}
\usepackage{mathrsfs}
\usepackage{subfigure}
\usepackage{enumitem}

\usepackage{amssymb}
\usepackage{amsmath}
\usepackage{amsthm}
\usepackage{mathtools}
\usepackage[utf8]{inputenc}

\usepackage{hyperref}
\usetikzlibrary{shapes.misc, positioning}
\usetikzlibrary{quantikz}
\usetikzlibrary{cd}

\usepackage{enumitem}
\setlist{leftmargin=0mm}

\tikzset{
operator/.append style={fill=black!5}}

\newcommand{\beq}{\begin{equation}}
\newcommand{\eeq}{\end{equation}}
\newcommand{\NE}[1]{\Delta S_{\hat N}({#1})}
\newcommand{\NA}[1]{#1|_{\hat N_A}}
\newcommand{\N}[1]{#1|_{\hat N}}
\newcommand{\Ket}[1]{\lvert#1\rangle}
\newcommand{\Bra}[1]{\langle #1\lvert}
\newcommand{\pu}[1]{\Pi_{#1}}

\newcommand{\sep}{\mathscr{D}^{\text{sep}}}
\newcommand{\dn}{\mathscr{D}_{\hat N}} 
\newcommand{\sepn}{\mathscr{D}^{\text{sep}}_{\hat N}}

\newcommand{\symsepn}{\mathscr{D}^{\textup{symsep}}_{\hat N}}
\newcommand{\dnloc}{\mathscr{D}_{\hat N_{\textup{local}}}}
\newcommand{\sepnloc}{\mathscr{D}^{\textup{sep}}_{\hat N_{\textup{local}}}}
\newcommand{\symsepnloc}{\mathscr{D}^{\textup{symsep}}_{\hat N_{\textup{local}}}}
\newcommand{\symsepnloci}{\mathscr{D}^{\textup{symsep}}_{\hat N_{i,\textup{local}}}}
\newcommand{\symbisepn}{\mathscr{D}^{\textup{sym-2-sep}}_{\hat N}}
\newcommand{\bisepn}{\mathscr{D}^{\text{2-sep}}_{\hat N}}

\newtheorem{lemma} {Lemma}
\newtheorem{proposition} {Proposition}
\newtheorem{corollary} {Corollary}
\theoremstyle{definition}

\newtheorem{conjecture}{Conjecture}

\begin{document}
\title{The typicality of symmetry-induced entanglement}

\author{Christian Boudreault}
\email{christian.boudreault@cmrsj-rmcsj.ca}
\affiliation{Coll\`ege militaire royal de Saint-Jean,
15 Jacques-Cartier Nord, Saint-Jean-sur-Richelieu, QC, Canada, J3B 8R8}

\author{Nicolas Levasseur}
\email{nicolas.levasseur@umontreal.ca}
\affiliation{Universit\'{e} de Montr\'{e}al, C.\,P.\,6128, Succursale Centre-ville, Montr\'eal, QC, Canada, H3C 3J7}

\date{\today} 

\begin{abstract}
\vspace{-1pt}
\begin{center}\textbf{\abstractname}\end{center}\vspace{-5pt} 
\noindent In the presence of a globally conserved charge $\hat N$, a natural question is whether a given separable state can be separated into charge-conserving components. We dub this problem the Symmetric Separability Problem (SSP). On random states, the SSP is answered negatively with probability one for almost all $\hat N$. Using a witness to the failure of symmetric separability, the \emph{number entanglement} (NE) introduced in~\cite{Ma2022symmetric}, we show that most symmetric and separable states are actually far from being symmetrically separable, with the NE featuring Gaussian concentration around a strictly positive mean value. We discuss some consequences of our results for quantum tasks in the presence of a superselection rule or in the absence of a common reference frame. Progress is made on the question of the size of the separable space constrained by $\hat N$. We also touch upon the question of the complexity of SSP, and multiparty entanglement. 
\end{abstract}

\maketitle

\section{Introduction}
The existence of nonclassical correlations might be the most profound outcome of the quantum revolution, both conceptually and practically. Entanglement, nonlocality and contextuality have advanced our understanding of matter and information~\cite{Zeng2015quantum, Nielsen_Chuang2010quantum}, and pushed back the limits of computational and communication tasks that can be realistically implemented~\cite{Bennett1984quantum,Bennett1993teleporting, Shor1994algorithms, Grover1996fast, Lloyd1996universal, Bouwmeester1997experimental, Raussendorf2013contextuality, Howard2014contextuality}. Finding which quantum states possess nonclassical correlations, quantifying the amount present, and characterizing the type of correlations is thus an urgent, but difficult endeavor. The outstanding Quantum Separability Problem (QSP) seeks to decide whether a given density matrix $\rho$ of a two-party system $\mathscr{H}_A \otimes \mathscr{H}_B$ can be written in the form
\beq\label{E:sep_main}
\rho=\sum_i p_i \rho_{A,i}\otimes \rho_{B,i},
\eeq
where $\rho_{A,i}$ and $\rho_{B,i}$ are densities on $A$ and $B$, respectively. The state is \emph{separable} and has only classical correlations if it can be written in the form~\eqref{E:sep_main}, and is entangled otherwise. As simple as it looks, the Quantum Separability Problem is known to be NP-hard in many cases, and believed to be so in general~\cite{Gurvits2003classical, Gharibian2010strong}. In the presence of a globally conserved charge, the natural question becomes whether a state can be separated into charge-conserving components. For a charge $\hat N$ such that $[\rho,\hat N]=0$, the separable state~\eqref{E:sep_main} is \emph{symmetrically separable for} $\hat N$ (\emph{symsep} for short) if 
\beq\label{E:symsep_main}
[\rho_{A,i}\otimes \rho_{B,i} , \hat N]=0 \qquad , \qquad \forall i.
\eeq
If $\rho$ is formally separable in the sense of~\eqref{E:sep_main} but not symmetrically separable in the sense of~\eqref{E:symsep_main}, it contains entanglement in at least one charge sector, and this entanglement will be unlocked by charge measurement~\cite{Ma2022symmetric}. We call this type of entanglement \emph{symmetry-induced} because prior to a choice of symmetry, and subsequent measurement, the separable state~\eqref{E:sep_main} contains only classical correlations. To the best of our knowledge, the Symmetric Separability Problem (SSP) --- deciding whether a given separable and symmetric state is symmetrically separable --- has no known general solution. In fact, SSP boils down to QSP within the restricted class $\dnloc$ of $\hat N_{\textup{local}}$-symmetric states, where $\hat N_{\textup{local}}$ is a localized version of $\hat N$. Within this class, separability and symmetric separability coincide (Fig.~\ref{F:sets_intro}). Thus when $\dnloc$ is of measure zero in $\dn$ (the set of $\hat N$-symmetric states), separable and symmetric states fail to be symmetrically separable with probability one. 

To cope with the algorithmic complexity of finding a general solution to the problem of entanglement, a complementary line of attack is to adopt a measure-theoretic point of view aiming at quantifying \emph{how much} entanglement is typical in physical states. Such knowledge is of value for theoretical investigations as well as numerical simulations, where one is led to make assumptions about typical quantum states or make use of statistical ensembles of states which presuppose minimal prior knowledge about a quantum system~\cite{Zyczkowski2001induced}. Celebrated findings from this approach include the typicality of nonseparable mixed states, along with a tradeoff between purity and separability~\cite{Zyczkowski1998volume1,Zyczkowski1999volume2}, the typicality of highly entangled pure states, and the existence of large subspaces in which all pure states are close to maximally entangled~\cite{Hayden2006aspects}. The toolbox of measure theory comprises a varied supply of theorems collectively known as \emph{concentration of measure}, the most familiar example being Lévy's Inequality Lemma, which states informally that any function of bounded gradient on a sphere will strongly concentrate around its mean value~\cite{Ledoux2001concentration}. Because (normalized) pure states and (normalized) purifications of mixed states live on spheres, the lemma will apply to functions of bounded gradient on these states. The entanglement entropy (EE), a witness of pure-state entanglement, is a notable example of the phenomenon~\cite{Hayden2006aspects}. With overwhelming probability, a random pure state will have an EE almost equal to the (near maximal) average pure-state EE, and will be (highly) entangled.  

To further investigate symmetric separability, we apply the concentration of measure phenomenon to another entanglement witness, the \emph{number entanglement} introduced in~\cite{Ma2022symmetric}, which is specifically tailored to diagnose the presence of symmetry-induced entanglement. For degenerate $\hat N$, we show that on $\hat N$-symmetric and separable states (i.e. on the set $\sepn$ of Fig.~\ref{F:sets_intro}) the number entanglement strongly concentrates around its (strictly positive) mean value. Thus, typical states in $\sepn$ will be found at some positive distance from symmetrically separable ones. (We will attempt to numerically evaluate that distance.) To emphasize, the concentration result shows that not only are typical states symmetrically inseparable, but also \emph{generically far} from being symmetrically separable. Thus, while it is a contribution to the Symmetric Separability Problem, the result is logically independent (i.e.~neither proves the other) from a general algebraic solution to this problem.  
\begin{figure}
\includegraphics{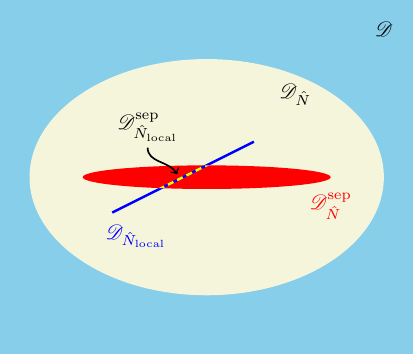}
\caption{\textbf{Pictorial representation of some classes of states considered in this work.} Class $\mathscr{D}$ contains all states on $\mathscr{H}_A \otimes \mathscr{H}_B$. A lower index $\hat N$ indicates symmetries, and an upper index `sep' or `symsep' indicates that states are separable or symmetrically separable, respectively. We algebraically show that $\sepnloc = \symsepnloc = \symsepn$. For a degenerate $\hat N$, $\sepnloc$ is of measure zero in $\sepn$. Thus, separable and $\hat N$-symmetric states fail to be \emph{symmetrically} separable with probability one. On the other hand, we show that $\text{vol}\sepnloc / \text{vol}\dnloc$ is nonzero. As the figure suggests, $\sepn$ is of positive measure in $\dn$ when $\hat N = \hat N_A\otimes 1_B + 1_A \otimes \hat N_B$.}\label{F:sets_intro}
\end{figure}

The paper is organized as follows. In Section~\ref{S:reduction}, SSP is reduced to QSP within the class $\dnloc$, and complexity of SSP is briefly touched upon. In Section~\ref{S:EE} we recall the phenomenon of concentration for the bipartite EE in pure states~\cite{Hayden2006aspects}. In Section~\ref{S:NE} we present the number entanglement (NE) and its property of being a witness to symmetric inseparability~\cite{Ma2022symmetric}. In Section~\ref{S:result} we show that the NE concentrates on $\sepn$, our main result, and explain what it tells us about the Symmetric Separability Problem. In Section~\ref{S:discussion} we discuss the relevance of our results, and Section\ref{S:conclusion} contains concluding remarks. Detailed proofs can be found in the Appendix. 

\section{Reduction to separability}\label{S:reduction}
The closed, convex set of density matrices (states) on a Hilbert space $\mathscr{H}$ will be denoted $\mathscr{D}(\mathscr{H})$, or simply $\mathscr{D}$ when the context is clear. We consider two-party states in $\mathscr{H}_A \otimes \mathscr{H}_B$, i.e. $\rho\in\mathscr{D}(\mathscr{H}_A \otimes \mathscr{H}_B)$, where  $\mathscr{H}_A$ and  $\mathscr{H}_B$ have respective dimensions $d_A$ and $d_B$. For a set of Hermitean operators $\hat N = \{\hat n_j\}$ on $\mathscr{H}_A \otimes \mathscr{H}_B$, we define the $\hat N$-symmetric states
\beq
\dn = \{\rho\in\mathscr{D}\colon [\rho,\hat n_j]=0\text{ for all } \hat n_j \in \hat N\}.
\eeq
In the following, we will consider the observables $\hat N = \{\hat N_A \otimes \hat N_B\}$, and $\hat N = \{\hat N_A \otimes 1_B + 1_A \otimes \hat N_B\}$, as well as their common `localization': 
\beq
\hat N_{\textup{local}} = \{\hat N_A \otimes 1_B, 1_A \otimes \hat N_B\}.
\eeq
Then $\symsepn$ will denote the set of symmetrically separable states, satisfying~\eqref{E:sep_main} and~\eqref{E:symsep_main} with the appropriate $\hat N$. 
It is an easy exercise to show
\begin{proposition}\label{T:localized_symsep}
$\symsepn = \symsepnloc$.
\end{proposition}
The proof consists in showing that Eqn.~\eqref{E:symsep_main} readily implies $[\rho_{A,i}\otimes \rho_{B,i} , \hat N_A \otimes 1_B]=0 =[\rho_{A,i}\otimes \rho_{B,i} , 1_A \otimes \hat N_B]$. Thus, SSP within $\dn$ is reduced to SSP within $\dnloc$. Let us define the twirl
\beq
\mathcal{G}[\rho]\coloneqq \int_{\mathcal{T}^2} dg\; T (g)\rho \, T^{\dagger}(g),
\eeq
where $T(g)=e^{ig_1 \hat N_A\otimes 1_B + ig_2 1_A \otimes \hat N_B}$, and $dg$ is the Haar measure over the two-dimensional torus $\mathcal{T}^2$. It follows from translation-invariance of the Haar measure that $\mathcal{G}[\rho]$ commutes with both $\hat N_A\otimes 1_B$ and $1_A \otimes \hat N_B$. The range of $\mathcal{G}$ is thus in $\dnloc$. 
Actually, 
\begin{proposition}
$\mathcal{G}$ is the (unique) retraction of $\mathscr{D}$ onto $\dnloc$, i.e. it collapses $\mathscr{D}$ onto $\dnloc$, and $\dnloc$ is the set of fixed points of $\mathcal{G}$. 
\end{proposition}
Because each $T(g)$ is of the form $e^{ig_1 \hat N_A}\otimes e^{ig_2 \hat N_B}$, it follows by linearity that $\mathcal{G}$ sends separable states to ($\hat N_{\textup{local}}$-symmetric) separable states. In fact, we have the stronger result
\begin{proposition}
$\sepnloc = \mathcal{G}[\sep] = \symsepnloc$.
\end{proposition}
\begin{corollary}\label{T:symsepn_sepnloc}
$\symsepn = \sepnloc$.
\end{corollary}
Therefore, SSP for $\hat N$ is reduced to QSP within the class of $\hat N_{\textup{local}}$-symmetric states. Importantly, when $\hat N$ is degenerate and not a multiple of the identity,  $\sepnloc$ is of measure zero in $\sepn$, and states in $\sepn$ fail to be symsep with probability one. Because of the SSP-to-QSP reduction, it is tempting to believe that for some symmetries at least, the SSP problem is NP-complete, as it is for the `trivial' case $\hat N = 1$~\cite{Gurvits2003classical, Gharibian2010strong}. An interesting question is: For which $\hat N$ is SSP an NP-complete problem? Progress towards an answer might come from knowledge of the relative size of $\sepn$ within $\dn$ (more specifically, of $\sepnloc$ within $\dnloc$). One can show that $\dn$ and $\sepn$ naturally split into the convex hulls of their charge sectors as follows:
\beq
\begin{aligned}\label{E:conv_main}
\dn &= \textup{conv}\bigcup_{n}\mathscr{D}(h_A^{n} \otimes h_B^{n})\\
\sepn &= \textup{conv}\bigcup_{n}\sep(h_A^{n} \otimes h_B^{n}),
\end{aligned}
\eeq
where charge-$n$ densities act on subspaces of the form $h_A^{n} \otimes h_B^{n}$ with $h_A^{n} \subset \mathscr{H}_A$, $h_B^{n} \subset \mathscr{H}_B$. Numerical evidence~\cite{Zyczkowski1998volume1, Zyczkowski1999volume2} suggests that each ratio 
\beq\label{E:ratio_main}
\frac{\textup{vol}\,\sep(h_A^{n} \otimes h_B^{n}) }{ \textup{vol}\,\mathscr{D}(h_A^{n} \otimes h_B^{n})}
\eeq
 is exponentially small in the dimension of $h_A^{n} \otimes h_B^{n}$. Because charge sectors are orthogonal, we conjecture that at least when all sector dimensions are large the relative size of $\sepn$ within $\dn$ will be exponentially small (in some combination of sector dimensions).

To conclude this section, let us mention that when $\hat N$ is \emph{nondegenerate}, it follows that $\sepn = \symsepn$, i.e.~SSP and QSP coincide. This is because any $\hat N$-symmetric state is diagonalized in the \emph{unique} eigenbasis of $\hat N$ (also an eigenbasis of $\hat N_A\otimes 1_B$ and $1_A \otimes \hat N_B$), whence $\dn=\dnloc$. Corollary~\ref{T:symsepn_sepnloc} in turn implies $\sepn = \symsepn$. Another extreme case is when $\hat N_A \otimes 1_B$ and/or $1_A \otimes \hat N_B$ is \emph{completely} degenerate, i.e.~a multiple of the identity, for $\hat N$ commutes with it. Then again, SSP and QSP coincide. Nontrivial behavior of the NE is to be searched for in between these extremes.

\section{Concentration}\label{S:concentration}
In this section, we show that the NE concentrates on $\sepn$, our main result, and explain what it tells us about the Symmetric Separability Problem. We begin by recalling facts about the concentration of entanglement entropy, some of which will be used later on.

\subsection{Concentration for pure-state entanglement entropy}\label{S:EE}
The von Neumann entropy of a state $\rho$ (pure or mixed) is 
\beq\label{E:vN_main}
S(\rho) = -\text{Tr}(\rho\log\rho).
\eeq
As usual, for  two-party states $\rho\in\mathscr{D}(\mathscr{H}_A \otimes \mathscr{H}_B)$ the bipartite entanglement entropy (EE) on $A$ is obtained by first tracing out $B$, and computing the von Neumann entropy of the resulting reduced matrix,
\beq
S(\rho_A) = S(\text{Tr}_B \rho).
\eeq
On pure states, the EE is sensitive only to nonclassical correlations between $A$ and $B$, and symmetric under the interchange of $A$ and $B$. We will write $\psi = \Ket\psi\Bra\psi$ for the density corresponding to a pure state $\ket\psi$. Normalized pure states live on the sphere $S^{2d_A d_B-1}$ up to a global phase. As was proved in~\cite{Hayden2006aspects}, pure-state EE is Lipshitz continuous: 
\begin{lemma}\label{T:hayden_main}
For pure states $\Ket\psi , \Ket\phi\in\mathscr{H}_A \otimes \mathscr{H}_B$, and $d_A \geq 3$, 
\beq
|S(\psi_A) - S(\phi_A) | \leq \eta \lVert \Ket\psi - \Ket\phi \rVert_2 ,
\eeq
with $\eta = \sqrt{8}\log d_A$, and where $\lVert\cdot\rVert_2$ is the Euclidean norm.
\end{lemma}
The most celebrated instance of measure concentration, Lévy's Inequality Lemma, deals specifically with Lipshitz functions on spheres:
\begin{lemma}[Lévy's Lemma]\label{T:levy_main}
Let $f : S^{n}\to \mathbb{R}$ be a Lipshitz function with constant $\eta$, i.e.
\beq
|f(x)-f(y)|\leq \eta\lVert x-y\rVert_2, \hspace{.8cm}\forall x,y \in S^{n},
\eeq
and let $X$ be a random variable on $S^n$ equipped with the Haar measure. Then  
\beq
\textup{Prob}(|f(X)-\mathbb{E}f(X)| > \alpha) \leq 2e^{-c(n+1)\alpha^2 / \eta^2},
\eeq
where $\mathbb{E}f(X)$ is the expectation value of $f(X)$, and $c$ is a positive constant that may be chosen as $c=(18\pi^3)^{-1}$. 
\end{lemma}
An immediate consequence of these lemmas is that the EE of a Haar-random pure-state will, with overwhelming probability, be nearly equal to the average pure-state EE. Together with the provable fact that the average EE is close to the maximal value $\log d_A$ when $d_A \ll d_B$, one concludes that \emph{generic} pure states have near-maximal bipartite entanglement~\cite{Hayden2006aspects}.

\subsection{Number entanglement}\label{S:NE}
It is a very general observation that failure to comply to a symmetry opens up the possibility of creating correlations by destroying quantum coherences with the use of a measurement or projection to a charge sector. The nonselective measurement of a single observable $\hat N$, a measurement for which the outcome is not recorded, will result in the state $\N \rho= \sum_{N} \Pi_{N}\rho\Pi_{N}$, where $\Pi_{N}$ is the projector onto the sector of charge $N$. Clearly, a nonzero von Neumann entropy \emph{difference} $S(\N\rho) - S(\rho)>0$ implies $[\rho, \hat N]\neq 0$. Violation of the symmetry is key to the measurement-induced generation of correlations.  

This idea is applied to the problem of symmetric separability as follows. We consider separable states~\eqref{E:sep_main} with a globally conserved charge, $[\rho , \hat N]=0$, and ask whether the state satisfies symmetric separability, Eqn.~\eqref{E:symsep_main}. If not, then it should be possible to increase the state's entropy by destroying quantum coherences through measurement. As before, we treat only the case of observables of the form $\hat N_A \otimes \hat N_B$ or $\hat N_A \otimes 1_B + 1_A \otimes \hat N_B$, with their localization written $\hat N_{\textup{local}}$. The \emph{number entanglement} (NE) of $\rho$ with respect to $\hat N$ was defined in~\cite{Ma2022symmetric} as
\beq\label{E:NE_main}
\NE\rho = S(\rho|_{\hat N_A \otimes 1_B}) - S(\rho).
\eeq
By construction, it measures the increase in von Neumann entropy produced by a local measurement of $\hat N_A$. (From now on, we follow the common practice of simply writing $\hat N_A$ for $\hat N_A \otimes 1_B$.) Note that the NE is invariant under the exchange of $A$ and $B$. Also note that on pure states, $\NE\phi = S(\NA\phi)$. The next proposition shows that a nonzero NE witnesses the failure of symmetric separability. 
\begin{lemma}\label{T:Ma_main} 
Symmetric separability of $\rho$ implies $\NE{\rho}=0$.
\end{lemma}
The proof is given in the Appendix, following the original proof in~\cite{Ma2022symmetric}, and consists in showing that symmetric separability implies that $\rho$ is already block-diagonal in eigenbases for $\hat N_A \otimes 1_B$ before any measurement is performed, so $S(\rho|_{\hat N_A} ) = S(\rho)$, and no entropy is created by the measurement. Actually, a stronger result is immediately obtained by construction of the NE:
\begin{proposition}\label{T:NE_zero_on_dnloc}
$\NE\rho\equiv 0$ on $\dnloc$.
\end{proposition}
The previous lemma follows since $\symsepn$ is a subset of $\dnloc$. 

\subsection{Concentration of the number entanglement}\label{S:result}
In this section, we determine how much symmetry-induced entanglement is typical among separable $\hat N$-symmetric states for degenerate $\hat N$. Our goal is to derive a concentration result for the NE in the set $\sepn$. To that end, purifying mixed states will allow us to apply the Lipshitz continuity of pure-state EE, Lemma~\ref{T:hayden_main}, and Levy's concentration inequality, Lemma~\ref{T:levy_main}, to the case of mixed-state NE, Eqn.~\ref{E:NE_main}. The results will rely on the construction of a probability measure $\mu$ on $\sepn$, and the partition of $\sepn$ into spherical shells equipped with the Haar measure. The construction depends on the realization of a manifold of purifications for $\sepn$, the \emph{purifying manifold} $\chi_0$, and a map $\tau\colon \chi_0 \to \sepn$. (See Figure~\ref{F:mappings}.) The purifying manifold is a smooth manifold within $(\mathscr{H}_A \otimes \mathscr{H}_B)^{\otimes 2}$, whose dimension $D$, also denoted $\text{p-dim }\sepn$, is called the \emph{purification dimension}. The purifying manifold $\chi_0$, which is assumed triangulable for simplicity~\cite{Boothby1975introduction}, is covered by a finite collection of domains of integration $\{\mathcal{B}_i\}_{i=1,...,M}$, where each $\mathcal{B}_i$ is diffeomorphic to a closed $D$-dimensional real ball. Each $\mathcal{B}_i$, with its pulled-back Euclidean measure, can in turn be split into spherical shells $\mathcal{S}_{ik}$ where the restricted Euclidean measure gives rise to the Haar measure. The measure $\tilde\mu$ on $\chi_0$ is finally pushed forward to $\sepn$,
\beq
\mu (S) = \tilde\mu(\cup_{\rho\in S} \;\tau^{-1}(\rho)).
\eeq
\begin{figure}
\includegraphics{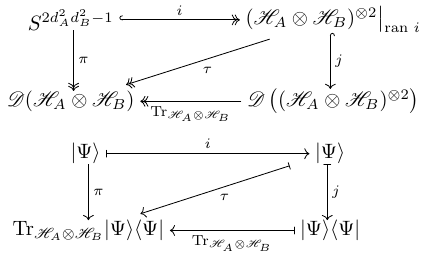}
\caption{\textbf{Sets and maps for the construction of the purifying manifold.} $i$ is a $C^{\infty}$-diffeomorphism, $j$ is $C^{\infty}$ injective, while $\pi, \tau$, and $\text{Tr}_{\mathscr{H}_A \otimes \mathscr{H}_B}$ are $C^{\infty}$ surjective, but not injective. $(\cdot)\big\lvert_{\text{ran } i}$ means restriction to the range of $i$. Rigorously speaking, the map $i$ is not from the sphere but from the complex projective space. This subtlety makes no difference in the argument.}\label{F:mappings}
\end{figure}
As before, we consider observables of the form $\hat N_A \otimes \hat N_B$ or $\hat N_A \otimes 1_B + 1_A \otimes \hat N_B$. The next proposition, an extension of Lemma~\ref{T:hayden_main}, shows that the NE of (purified) densities is of bounded gradient on the purifying manifold:
\begin{proposition}\label{T:lipschitz_main}
When $\mathscr{H}_A\otimes\mathscr{H}_B$ has dimension $d_A d_B\geq 3$, the NE is Lipshitz continuous on $\mathscr{D}$, i.e.
\beq
|\NE\sigma - \NE\rho|\leq \eta \big\lVert \Ket{\pu{\sigma}}- \Ket{\pu{\rho}}\big\rVert_2,
\eeq
with $\eta = 4\sqrt{2}\log (d_A d_B)$, and where $\Ket{\pu{\sigma}}, \Ket{\pu{\rho}}$ are any purifications of $\sigma, \rho$ on a common space. Specializing to $\sepn$ and its purifying manifold, we have : 
\beq
|\NE\sigma - \NE\rho|\leq \eta \big\lVert \Ket{\pu{\sigma}}- \Ket{\pu{\rho}}\big\rVert_2^{(D)},
\eeq
with $D=\textup{p-dim }\sepn$, $D$-dimensional Euclidean norm $\lVert \cdot\rVert_2^{(D)}$, and $\eta = 4\sqrt{2}\log D$. Also, $\Ket{\pu{\sigma}},\Ket{\pu{\rho}}$ are any purifications of $\sigma,\rho$ in a common domain of integration $\mathcal{B}_i$.
\end{proposition}
It then follows from Lemma~\ref{T:levy_main} that
\begin{proposition}\label{T:conc_NE_conditional_main}
Separable states that are symmetric for $\hat N$ present a strong concentration of NE around the mean $m_{ik}$ in each spherical shell of the purifying manifold: for $\rho\in \tau(\mathcal{S}_{ik})$,
\beq
\textup{Prob}_{\mu}(| \NE\rho - m_{ik} | > \alpha) \leq 2v_{ik}e^{-c D \alpha^2/\eta^2},
\eeq
where $m_{ik}=v_{ik}^{-1}\int_{\tau(\mathcal{S}_{ik})}d\mu(\sigma)\NE\sigma$, and $D$ is the purification dimension.
\end{proposition}
The NE strongly concentrates around its local mean $m_{ik}$ on each shell. Due to (Lipshitz) continuity, the NE must agree where shells meet, forcing local means $m_{ik}$ to coalesce around the global mean $m$. We thus obtain
\begin{proposition}\label{T:genN_main}
For an observable $\hat N$ with large enough purification dimension $D$, the NE on $\sepn$ strongly concentrates around its mean value:
\beq
\textup{Prob}_{\mu}(| \NE\rho - m | > \alpha) \lesssim O(e^{-c D \alpha^2 /\eta^2}),
\eeq
where $m$ is the (strictly positive) mean NE of $\hat N$ averaged over all states in $\sepn$, $c$ is a positive constant that may be chosen as $c=(18\pi^3)^{-1}$, and $\eta = 4\sqrt{2}\log D$.
\end{proposition}
In the above inequality, we have dropped a term that becomes negligible at (moderately) large values of $D$. (See details in the Appendix.) Therefore, our result applies to all degenerate observables of the form $\hat N_A \otimes \hat N_B$ or $\hat N_A \otimes 1_B + 1_A \otimes \hat N_B$, except those with small purification dimension $\textup{p-dim }\sepn$, i.e. operators of lower dimension or operators embodying stringent symmetries with few symmetric states and/or symmetric states that are mostly inseparable. 

Proposition~\ref{T:genN_main} reveals the existence of a `statistical distance' between the sets $\sepn$ and $\symsepn$, in the sense that separable and symmetric states strongly cluster at a finite distance from $\symsepn$. States that are close to being symmetrically separable are statistically hard to find. The concentration sharpens overwhelmingly with increasing (purification) dimension. Numerics for a few special case of physical interest validate the finding, showing the narrowing of states around the average NE value as the total dimension of the system is increased. Fig.~\ref{F:numerics_qudits_sum} displays the case $\hat N = \hat N_A \otimes \hat N_B$ where each subsystem $A$, $B$ support a single qudit, and $N_A$, $N_B$ are nondegenerate, i.e.~single-qudit level numbers. (See Appendix, Fig.~\ref{F:numerics_qudits_prod}, for the two-qudit case  $\hat N_A \otimes 1_B + 1_A \otimes \hat N_B$, and Fig.~\ref{F:numerics_qubits_sum_prod} for multi-qubit states where $N_A$, $N_B$ each provide $\sum \sigma_z$ on subsystems.) The mean NE is observed to decrease moderately fast with increasing dimension, while standard deviation decreases exponentially fast, guaranteeing swift concentration. Low-dimensional numerics already seem to point to a difference between qudit-based and qubit-based symmetric separability (with charges as above) when total dimensions are comparable, for the mean NE is always larger in the qubit case. This is consistent with a smaller relative volume for $\symsepn$ in the qubit case, in line with our conjecture following expressions \eqref{E:conv_main}, \eqref{E:ratio_main}. Indeed, a $q$-qubit system has $O(q)$ $\sigma_z$-sectors while the qudit-based system of the same size has $O(2^{q/2})$, which thus tend to be much smaller. 

We conclude this section by mentioning that Proposition~\ref{T:genN_main} gives un upper bound on the probability distribution around the mean NE, which can be turned into an upper bound for the purification dimension $D=\text{p-dim }\sepn$ by numerically sampling the probability distribution. Our numerics indicate, however, that Proposition~\ref{T:genN_main} is very far from being saturated at lower dimensions, with an observed concentration that is tremendously faster than it would be even with a maximum purification dimension $D_{\text{max}}=2d_A^2 d_B^2$. We do not know if Proposition~\ref{T:genN_main} is closer to optimum at higher dimensions.
\begin{figure}
\includegraphics[width=0.5\textwidth]{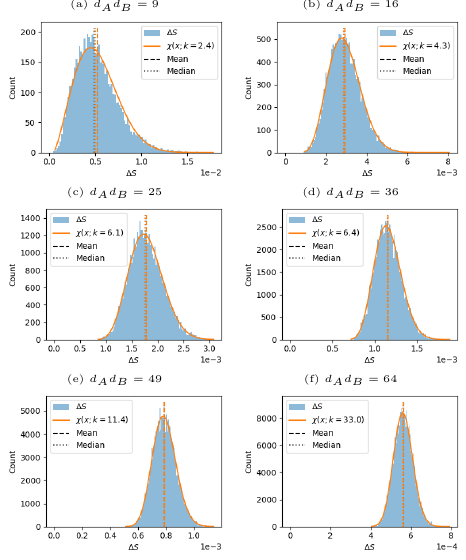}
\caption{Distribution of NE values for randomly generated 2-qudit states in $\sepn$, showing concentration around the mean as dimension increases. Each subsystem $A$, $B$ supports a single qudit with $d \in \{3,...,8\}$, and $\hat N = \hat N_A \otimes \hat N_B$ with $\hat N_A$, $\hat N_B$ possessing nondegenerate eigenvalues $\{0,1,...,d-1\}$, i.e.~single-qudit level number on each subsystem. The distributions are fitted to an empirical chi distribution of order $k$ that depends on dimension.}
\label{F:numerics_qudits_sum}
\end{figure}

\section{Discussion}\label{S:discussion}
\subsection{Symmetric separability and superselection}\label{S:SSR}
Symmetrically separable (symsep) states have measure zero among separable $\hat N$-symmetric states (for degenerate $\hat N \neq 1$). This is of special importance when the symmetry belongs to a superselection rule, especially when it arises from a lack of shared reference frame (RF)~\cite{Bartlett2007reference}. Indeed, let Jane's and John's RFs be related by the transformation $g\in G$, for a compact group $G$.\footnote{We refrain from using the familiar characters Alice and Bob to avoid any confusion with the subsystems $A$ and $B$. Jane and John will be describing the \emph{same} system $\mathcal{H}=\mathcal{H}_A \otimes \mathcal{H}_B$ from their respective points of view.} If Jane's RF is uncorrelated with John's, not only is the instantaneous value of $g$ unknown to John, but he must avoid any bias in reconstructing Jane's state $\rho$ in his own frame. He will therefore assume the state to be given, in his frame, by the unbiased twirl
\beq\label{E:twirl}
\mathcal{G}[\rho] = \int_{G} dg\; T(g)\tilde\rho\, T^{\dagger}(g),
\eeq 
where $T=T_A \otimes T_B$ is a unitary representation of $G$ on $\mathcal{H}_A \otimes \mathcal{H}_B$. (Note that Jane and John lack a common RF, but they agree on the nature of the subsystems $A$ and $B$.) Then one can show that
\beq
[\mathcal{G}[\rho], T(g)]=0\quad , \quad \forall g\in G.
\eeq
By the same token, any POVM $\{E_k \}$ in Jane's frame should be considered as an unbiased POVM $\{\mathcal{G}[E_k]\}$ in John's frame, and it follows that
\beq
[\mathcal{G}[E_k], T(g)]=0\quad , \quad \forall g\in G.
\eeq
A similar but higher-level construction can be used for general quantum operations. A completely positivity-preserving superoperator $\mathcal{E}$ in Jane's frame should be considered by John as the unbiased superoperator
\beq\label{E:superop_JohnRF}
\tilde{\mathcal{G}}[\mathcal{E}] = \int_{G} dg\; \mathcal{T}(g)\circ \mathcal{E} \circ \mathcal{T}(g^{-1}),
\eeq 
where $\mathcal{T}(g)[\rho] = T(g)\rho T^{\dagger}(g)$ is a unitary representation of $G$ on $\mathscr{D}(\mathcal{H}_A \otimes \mathcal{H}_B)$. Then one finds
\beq\label{E:superop_comm}
[\tilde{\mathcal{G}}[\mathcal{E}], \mathcal{T}(g)]=0\quad , \quad \forall g\in G.
\eeq
We refer to~\cite{Bartlett2007reference} for the details of these constructions. It is now manifest that \emph{relative to John's RF}, all preparations, measurements, and operations that Jane can implement commute with unitary representations of $G$. In other words, all states, measurements, and operations act on a direct sum $\mathcal{H}_A \otimes \mathcal{H}_B=\oplus_n \mathcal{H}_n$, where each sector $\mathcal{H}_n$ carries an inequivalent unitary representation of $G$ with charge $n$. This situation is more commonly known as a \emph{superselection rule} (SSR). (Here it results from the lack of a shared RF.\footnote{Any superselection rule associated with unitary representations of a compact group results from the lack of an appropriate RF.~\cite{Bartlett2007reference}})

Now assume that Jane and John lack a common phase (e.g.~a shared atomic clock), and let the associated unitary representation be $T(g)=e^{ig\hat N}$, where $\hat N = \hat N_A \otimes 1_B + 1_A \otimes \hat N_B$.  Notice that $T(g)=T_A(g)\otimes T_B(g)$. Any state $\rho$ that Jane prepares commutes with $\hat N$ in John's frame. In fact, the twirling operation $\mathcal{G}$, Eqn.~\eqref{E:twirl}, is the unique \emph{retraction}\footnote{A retraction of $X$ onto $Y\subset X$ is a map $r\colon X \to Y$ such that $r(y)=y$ for all $y\in Y$.} of $\mathscr{D}$ onto $\dn$. Because symmetric states are the states fixed by $\mathcal{G}$, they are the ones that can be communicated without a shared RF, i.e.~\emph{fungible} states. In particular, $\sepn$ is pointwise fixed by $\mathcal{G}$, so these states are the separable states that can be communicated without a shared RF, i.e. fungible separable states. From John's point of view, however, these states are only formally separable since they are \emph{symmetry-induced} entangled with probability one --- and possess an NE close to the nonzero average NE with overwhelming probability. 
   
The truly local states, according to John, are the symsep states. These are the states he can produce, following instructions from Jane, out of \emph{fungible} product states and \emph{fungible} LOCCs. Indeed, let them start with the fungible (i.e.~symmetric) product state $\rho_A \otimes\rho_B \in \dn$, on which they agree: $\mathcal{G}[\rho_A \otimes\rho_B]=\rho_A \otimes\rho_B$. It is easy to verify that this state is actually in $\symsepn$. Let them apply any number of fungible LOCC rounds on it. A simple class of fungible LOCCs is the class of symmetric LOCCs, defined as follows. Instead of LOCCs we consider the strictly larger set $\mathsf{SEP}$ of \emph{separable} operations, having product-form Kraus operators,
\beq\label{E:SEP_ops}
\mathcal{E}(\rho) = \sum_j (K_A^j \otimes K_B^j)\rho (K_A^j \otimes K_B^j)^{\dagger},
\eeq
with $\mathcal{E}(1) = 1$. Separable operations have the defining property of sending separable states to separable states, and it is known that $\overline{\mathsf{LOCC}_{\mathbb{N}}}\subsetneq\mathsf{SEP} $~\cite{Macieszczak2019coherence}. We say that $\mathcal{E}\in\mathsf{SEP}_{\hat N}$ if $\mathcal{E}\in\mathsf{SEP}$ and   
\beq
[K_A^j \otimes K_B^j, \hat N]=0
\eeq 
for all $j$. (Symmetric LOCCs form the subset $\mathsf{LOCC}\cap\mathsf{SEP}_{\hat N}$.) When $\mathcal{E}\in\mathsf{SEP}_{\hat N}$ it is straightforward to show that it is fungible,
\beq
\tilde{\mathcal{G}}[\mathcal{E}]=\mathcal{E}.
\eeq
Also $\mathcal{E}(\rho_A \otimes\rho_B)$, or equivalently $\tilde{\mathcal{G}}[\mathcal{E}](\mathcal{G}[\rho_A \otimes\rho_B])$, is in $\sep$. It was shown in~\cite{Macieszczak2019coherence, Ma2022symmetric} that the NE does not increase under operations from $\mathsf{SEP}_{\hat N}$. Since $\NE{\rho_A \otimes\rho_B} = 0$ it follows that $\NE{\mathcal{E}(\rho_A \otimes\rho_B)} = 0$, so $\mathcal{E}(\rho_A \otimes\rho_B)\in\dnloc$. Being separable, we find that $\mathcal{E}(\rho_A \otimes\rho_B)$ is in $\symsepn$. This shows that symsep states are the truly local states according to John, those states he can produce following instructions from Jane out of fungible product states and fungible LOCCs. Equivalently, non-symsep states cannot be produced by John using only fungible product states and fungible LOCCs. If Jane stands for John's complement (a.k.a.~Nature), then in the presence of a SSR (or in the absence of an absolute RF) John should identify the symsep states as the only genuinely unentangled states of Nature. Let us also mention that $\symsepn$ is closed under the action of $\mathsf{SEP}_{\hat N}$. 

We expect that most unentangled states in Jane's frame ($\rho\in\sep$) will be categorized as entangled in John's, since $\symsepn$ is a nullset of $\mathcal{G}[\sep]$. It is interesting to turn the tables, and ask instead if Jane can find entangled (and necessarily nonsymmetric) states $\rho$ that are identified as unentangled by John, i.e.~states $\rho\in \mathcal{G}^{-1}[\symsepn]\backslash \sep$. This is answered positively, and examples are actually easy to find. Consider the 2-qubit pure state $\ket{\Psi}= a\ket{00} + \sqrt{1-|a|^2}\ket{11}$, $|a|\in (0,1)$, and $\rho = \ket{\Psi}\bra{\Psi}$, which is entangled. Let $\hat N_A = \hat N_B = \left(\begin{smallmatrix}0&\\&1\end{smallmatrix}\right)$ be the local number of ones in the computational basis, and 
\beq
\hat N = \hat N_A \otimes 1_B + 1_A \otimes \hat N_B = \left(\begin{smallmatrix}0&&&\\&1&&\\&&1&\\&&&2\end{smallmatrix}\right). 
\eeq
Then $\rho|_{\hat N} = |a|^2 \ket{0}\bra{0}\otimes \ket{0}\bra{0} + (1-|a|^2)\ket{1}\bra{1}\otimes \ket{1}\bra{1}$, which is symsep. The reason the state $\rho$ is disentangled by twirling is that the coherences between $\ket{00}$ and $\ket{11}$ are destroyed by projecting to the charge sectors of $\hat N$, which do not mix these states (charge 0 and 2, respectively). Genuine multipartite entanglement (GME) can also be destroyed by twirling. Consider a 3-qubit system with one qubit in $A$ and two qubits in $B$. Let again use the local number of ones,
\beq
\begin{aligned}
\hat N_A &= \left(\begin{smallmatrix}0&\\&1\end{smallmatrix}\right), \hat N_B= \left(\begin{smallmatrix}0&&&\\&1&&\\&&1&\\&&&2\end{smallmatrix}\right),\\
\hat N &= \left(\begin{smallmatrix}0&&&&&&&\\&1&&&&&&\\&&1&&&&&\\&&&2&&&&\\&&&&1&&&\\
&&&&&2&&\\&&&&&&2&\\&&&&&&&3\end{smallmatrix}\right).
\end{aligned}
\eeq
Because $\hat N$ does not have a common charge sector for $\ket{000}$ and $\ket{111}$, twirling by $\hat N$ will disentangle the GHZ-type states $\ket{\psi_{\text{GHZ}}}=a\ket{000}+\sqrt{1-|a|^2}\ket{111}$. It will not disentangle W-type states $\ket{\psi_{\text{W}}}=a\ket{001}+b\ket{010}+\sqrt{1-|a|^2 - |b|^2}\ket{100}$, since $\hat N$ has a common charge sector for these charge-1 kets, so $\rho_{\text{W}}|_{\hat N} = \rho_{\text{W}}$, where of course $\rho_{\text{W}} = \ket{\psi_{\text{W}}}\bra{\psi_{\text{W}}}$. A minor change is sufficient to disentangle W-states as well, however. Use for instance $\hat N'_A = \left(\begin{smallmatrix}0&\\&3\end{smallmatrix}\right)$. Then the charge sectors of $\hat N$ are modified to
\beq
\hat N' = \left(\begin{smallmatrix}0&&&&&&&\\&1&&&&&&\\&&1&&&&&\\&&&2&&&&\\&&&&3&&&\\
&&&&&4&&\\&&&&&&4&\\&&&&&&&5\end{smallmatrix}\right),
\eeq
decoupling $\ket{001}$ and $\ket{010}$ from $\ket{100}$. The result of twirling on a W-state is symsep:
\beq
\rho_{\text{W}}|_{\hat N'} = \ket{0}\bra{0}\otimes \ket{\Phi_{ab}}\bra{\Phi_{ab}}+
\xi_{ab}\ket{1}\bra{1}\otimes \ket{00}\bra{00},
\eeq
with $\ket{\Phi_{ab}}=a\ket{01}+b\ket{10}$ and $\xi_{ab}=\sqrt{1-|a|^2 - |b|^2}$. 

These simple examples are \emph{ad hoc}. For more general entangled states $\rho$, can we find the operators $\hat N$ that destroy entanglement by twirling so that $\mathcal{G}_{\hat N}[\rho]\in\symsepn$? Can this approach be used to further characterize the structure of multiparty entangled states? And does there exist \emph{symmetry-agnostic} entanglement, resisting twirling for \emph{all} $\hat N$, and how common is it? We leave these questions for future work.

\subsection{Symmetry-induced entanglement as resource}\label{S:symsep_vs_computation}
Symsep states and symmetric LOCCs can be considered as free states and operations in certain contexts, especially in the absence of a shared reference frame where they represent fungible information that can be created locally and communicated meaningfully with no additional resource (e.g.~a clock or a system of axes). In this context, entanglement will correspond to a (nonlocal, nonfungible) resource, and will allow for the realization of tasks beyond the capabilities of free states alone, e.g.~local distinguishability of quantum states, and teleportation~\cite{Verstraete2003quantum, Schuh2004nonlocal}. Regular entanglement ($\rho\in\dn \backslash \sepn$) and symmetry-induced entanglement ($\rho\in\sepn \backslash \symsepn$) both represent resources for these tasks, although with distinct behavior. Symmetry-induced entanglement is likely to be more robust than regular entanglement, however, since $\symsepn$ is a nullset in $\sepn$ which in turn has positive volume in $\dn$~\cite{Parez2024fate}.  

Surely, not all resources are equally useful for all tasks. It is well-known that most quantum states are \emph{too entangled} to be useful for quantum speedup in measurement-based quantum computing~\cite{Bremner2009random, Gross2009quantum}. This is explained by the fact that correlations in highly entangled states are well approximated stochastically, so these states cannot solve problems in the complexity class $\mathsf{BQP}\setminus\mathsf{BPP}$. In~\cite{Gross2009quantum} for instance, it was proven that an $n$-qubit pure state $\Ket{\Psi_n}$ cannot provide superpolynomial speedup if $E_g (\Ket{\Psi_n})=n-O(\log n)$, where $E_g$ is the geometric measure. A concentration result is then obtained for $n \geq 11$, 
\beq
\text{Prob}\big(E_g (\Ket{\Psi_n})<n-\Omega(\log n)\big)<e^{-n^2},
\eeq
showing that the fraction of universal resources among $n$-qubit pure states is less than $e^{-n^2}$. 

The phenomenon of symmetry-induced entanglement does not exist for pure states. For mixed states however, it is of practical interest to determine how symmetry-induced entanglement can be used as a resource for computation, and whether states with symmetry-induced entanglement generically have too little or too much of it to be useful. Because the NE is a symmetry-induced entanglement monotone~\cite{Ma2022symmetric}, Proposition~\ref{T:genN_main} implies that generic states in $\sepn$ possess a nonzero amount of symmetry-induced entanglement close to average with overwhelming likeliness. Numerical work suggests that this average amount is decreasing with increasing system dimension. See Figs.~\ref{F:numerics_qudits_sum},~\ref{F:numerics_qudits_prod},~\ref{F:numerics_qubits_sum_prod}.

\subsection{Multiparty entanglement}\label{S:multiparty}
When we consider multipartite systems with conserved total charge $\hat N = \sum_i \hat N_i$, we should adapt the usual definitions, and consider that a state is \emph{symmetrically partially separable} if it is symmetrically separable with respect to some fixed bipartition, and \emph{symmetrically biseparable} if a convex decomposition $\rho=\sum_i p_i \rho_i$ exists where each $\rho_i$ is symmetrically partially separable (with respect to a bipartition that may depend on $i$). A state will contain symmetry-induced \emph{genuine multipartite entanglement} (GME)~\cite{Seevinck2001sufficent} if it is not symmetrically biseparable. The proofs of our propositions can be partly generalized to those cases, as we now explain.

\emph{Partial and full separability.} Consider a system $\Omega$ with $K$ subsystems. Then $\rho$ is symmetrically partially separable over the fixed bipartition $A \cup B = \Omega$ if $\rho \in \symsepn (\mathscr{H}_{A}\otimes\mathscr{H}_{B})$. This is exactly the case treated already, and all our results apply word for word. The case of symmetrically fully separable states $\rho = \sum_i p_i \rho_1 \otimes \cdots \otimes \rho_K$, with termwise symmetry, is also easily covered. We immediately find that $\symsepn = \sepnloc$ (with an obvious redefinition of $\hat N_{\textup{local}}$), so that this set has zero measure in $\sepn$ when $\hat N$ is degenerate. The NE vanishes identically on $\dnloc$, and concentrates around its mean value on $\sepn$.  

\emph{Biseparability.} To treat the case of symmetrically biseparable states of system $\Omega$, it is convenient to define them formally over the space
\beq
\mathscr{H}_{\textup{bipartitions}}=\bigoplus_i \mathscr{H}_{A_i}\otimes\mathscr{H}_{B_i},
\eeq
where $\{(A_i , B_i)\}$ is the set of all bipartitions of $\Omega$, $A_i \cup B_i = \Omega$ for all $i$. A state $\rho$ is $\hat N$-symmetric and biseparable, $\rho\in\bisepn$, if
\beq
\rho \in \textup{conv}\bigcup_{i} \sepn (\mathscr{H}_{A_i}\otimes\mathscr{H}_{B_i}).
\eeq
A state is symmetrically biseparable, $\rho\in\symbisepn$, if
\beq
\rho \in  \textup{conv}\bigcup_{i} \symsepnloci (\mathscr{H}_{A_i}\otimes\mathscr{H}_{B_i}),
\eeq
where $\hat N_{i,\textup{local}}$ localizes $\hat N$ on the bipartition $(A_i , B_i)$. By the same argument as before, $\symsepnloci$ is a nullset in $\sepn$. Even if we lack a monotone for symmetric biseparability (that is, a generalization of the NE), our results already confirm concentration within each sector $\sepn (\mathscr{H}_{A_i}\otimes\mathscr{H}_{B_i})$. However, we \emph{cannot} conclude that $\symbisepn$ is of measure zero in $\bisepn$, because sectors might be nonorthogonal. In other words, a randomly chosen state in $\bisepn$ (i.e. symmetric and biseparable) might possess \emph{symmetry-induced} GME with probability less than one. The probability will depend on $\hat N$, more specifically on the degree of orthogonality of the sectors $\sepn (\mathscr{H}_{A_i}\otimes\mathscr{H}_{B_i})$, and their dimensions. See Eqn.~\ref{E:conv_main} and its explanation.

Although incomplete, these results are fully compatible, if somewhat complementary, to the fate of entanglement identified in~\cite{Parez2024fate} for generalized quantum evolutions to equilibrium (along temperature, time, and spatial separation) with and without the presence of superselection. By imposing extremely severe constraints on the possibilities of states to evolve towards more local, separated forms, symmetries delay the partial or complete disappearance of entanglement (e.g.~sudden death). Symmetry-constrained entanglement is more robust than its unconstrained counterpart. It is compelling to study symmetry-induced entanglement as a nonfungible resource for multiparty tasks~\cite{Verstraete2003quantum, Schuh2004nonlocal, Schuh2004quantum}. Is it possible, for instance, to extend to the symmetry-induced case the fact that GME-activatable states coincide with partially separable ones~\cite{Palazuelos2022genuine}? Another question beyond the scope of this article is the extent to which the decisive role of GME in quantum metrology~\cite{Hyllus2012Fisher} and quantum cryptography~\cite{Das2021universal} conveys to symmetry-induced GME. Our results might also help characterize the symmetries for which the set of symmetric biseparable states is so thin as to essentially preclude the sudden death of symmetry-constrained GME (e.g.~when the sectors $\sepn (\mathscr{H}_{A_i}\otimes\mathscr{H}_{B_i})$ are orthogonal). Systems would carry an extremely robust form of GME in the present of such symmetries. 

In a different vein, a duality between separability and purity was numerically detected in~\cite{Zyczkowski1998volume1}, and conjectured to hold in general. Entanglement is typical for pure states (pure separable states form a nullset~\cite{Popescu1994bell, Popescu1995bell}, and generic pure states are highly entangled~\cite{Hayden2006aspects}), while separability is more and more common as purity is lowered, becoming necessary below a certain threshold (in dimension $d$, $\text{Tr}(\rho^2)<1/(d-1)$ implies separability~\cite{Gurvits2002largest}). It would be interesting to determine if a relationship holds, within the subset $\sepn$, between purity and symmetric separability, given that purities 1 and $1/d$ both ensure symmetric separability in $\sepn$.

\subsection{Assumptions and scope}\label{S:scope}
Here we discuss some assumptions that were made, and the corresponding scope of our results.

\emph{The purifying manifold.} When using the purifying manifold of $\hat N$ in the proofs, we are assuming that it contains purifications for most states in $\sepn$. That is, we assume that lower-dimensional Euclidean pieces of $\chi$ --- see Appendix~\ref{S:conc_NE_appendix} --- as well as its non-Euclidean pieces are somehow negligible compared to the purifying manifold $\chi_0$. If however this was not the case, for instance if $\chi_0$ was vacuous, then our conclusions would be invalidated. Can we characterize the observables $\hat N$ for which the purifying manifold assumption is wrong, that is, for which the purifying manifold does not represent most states in $\sepn$ in any physically reasonable sense? We note however that when the assumption is satisfied, our results are independent of the choice of a particular triangulation (or particular finite open cover).

\emph{Specific $\hat N$.} We stress that our concentration result describes a structural property of operators that only have in common that they are degenerate, and have a well-defined purification dimension $\text{p-dim}\,\sepn$ which is not too small.  All operators `typical' in this sense fall within the scope of our propositions. The drawback of abstraction, however, is the difficulty to apply our results to a specific $\hat N$ without prior knowledge of $\text{p-dim}\,\sepn$. Establishing results for specific, physically motivated operators is the object of futur investigation.

\section{Conclusion}\label{S:conclusion}
In this work we have studied the Symmetric Separability Problem --- deciding whether a given separable and $\hat N$-symmetric state is symmetrically separable --- or, equivalently, regular separability within a restricted class of states with extended, `localized' symmetry. When $\hat N\neq 1$ is degenerate, separable and symmetric states fail to be symmetrically separable with probability one. 

We have used the NE, a monotone for symmetry-induced entanglement, to show that a bipartite state which is both $\hat N$-symmetric and separable will contain an amount of symmetry-induced entanglement which is superexponentially close to the (strictly positive) average amount on $\sepn$. Simply put, in the presence of a typical degenerate symmetry, generic separable states will be far from being symmetrically separable, and will contain a finite amount of symmetry-induced entanglement locked in their charge sectors. 

We have discussed the relevance of our results for computation and communication tasks, especially in the presence of a superselection rule, or in the absence of a shared reference frame. In this scenario, symmetrically separable states and symmetric LOCCs represent free fungible information and processes. States in $\symsepn$ are unique in being constructible out of (fungible) product states and (fungible) LOCCs, and should be considered as the only genuinely unentangled states in the absence of a shared frame, while all other states correspond to nonlocal resources.  

Finally, we have shown that our results are also of import to multiparty entanglement, and are consistent with recent findings on the fate of entanglement over general system evolution. Our work supports the view that symmetry-constrained entanglement is more robust than its unconstrained counterpart, and could help characterize systems with optimal GME resilience.

\begin{acknowledgments}
CB acknowledges support from the Canadian Defence Academy under Grant No.~102564, and support from the Institut Transdisciplinaire d'Information Quantique, and Fonds de Recherche du Québec. NL acknowledges financial support from his PhD advisor's research fund, Prof. Richard MacKenzie. CB wishes to thank William Witczak-Krempa for sparking his interest in the problem of symmetric separability, and for many fruitful discussions.\end{acknowledgments}

\section*{Author contributions}
CB contributed to ideation, theoretical development, numerics, interpretation of the results, and final drafting of the article. NL was responsible for the numerical work and analysis. 

\section*{Conflicts of interest}
The authors declare no conflict of interest.

\section*{Codes and data}
Codes used for the numerical work are available at \href{https://github.com/dslap0/Number-Entanglement}{github.com/dslap0/Number-Entanglement}.

\appendix
\onecolumngrid

\section{Detailed proofs}
\subsection{Reducing SSP to QSP}
A \emph{separable} state on $\mathscr{H}_A \otimes \mathscr{H}_B$ is a mixture of factored states, 
\beq\label{E:sep}
\rho=\sum_i p_i \rho_{A,i}\otimes \rho_{B,i},
\eeq
where $\rho_{A,i}$ and $\rho_{B,i}$ are densities on $A$ and $B$, respectively. For an observable $\hat N$ on $A\cup B$, we say that a separable state~\eqref{E:sep} is \emph{symmetrically separable for} $\hat N$ if 
\beq\label{E:symsep}
[\rho_{A,i}\otimes \rho_{B,i} , \hat N]=0 \qquad , \qquad \forall i.
\eeq
Obviously, symmetric separability for $\hat N$ implies $[\rho , \hat N]=0$. We denote the set of separable states as $\sep$, and the set of states symmetrically separable for $\hat N$ as $\symsepn$. In what follows, we exclusively consider observables of the form $\hat N_A \otimes \hat N_B$ or $\hat N_A \otimes 1_B + 1_A \otimes \hat N_B$. The next proposition shows that a nonzero NE witnesses the failure of symmetric separability. 
\begin{proposition}[Ma \emph{et al.}]\label{T:Ma} 
Symmetric separability of $\rho$ implies $\NE{\rho}=0$.
\end{proposition}
\begin{proof}
The proof for operators of the form $\hat N_A \otimes 1_B + 1_A \otimes \hat N_B$ is given in the original paper~\cite{Ma2022symmetric}. Here, we give only the (very similar) argument for the case $\hat N_A \otimes \hat N_B$. Symmetric separability means that $\rho$ is of the form~\eqref{E:sep} with termwise conservation~\eqref{E:symsep}. Thus
\beq
\begin{aligned}
\rho_{A,i}\otimes \rho_{B,i} &= \hat N^{\dagger}(\rho_{A,i}\otimes \rho_{B,i})\hat N\\
	&=\hat N_A^{\dagger}\rho_{A,i}\hat N_A \otimes \hat N_B^{\dagger}\rho_{B,i}\hat N_B\\
	&=\tilde\rho_{A,i}\otimes \tilde\rho_{B,i}
\end{aligned}
\eeq
for densities $\tilde\rho_{A,i},\tilde\rho_{B,i}$. Expanding the equations
\beq
\rho_{A,i}\otimes \rho_{B,i} -\tilde\rho_{A,i}\otimes \tilde\rho_{B,i}=0
\eeq
 in Pauli words, one shows that $\tilde\rho_{A,i} = k\rho_{A,i}$ and $\tilde\rho_{B,i} = \tilde\rho_{B,i}/k$ for some nonzero constant $k$. Unicity of the trace demands $k=1$. Hence, symmetric separability implies $[\rho_{A,i},\hat N_A]=0$ and $[\rho_{B,i},\hat N_B]=0$. (For operators of the form $\hat N_A \otimes 1_B + 1_A \otimes \hat N_B$, the same result is obtained by taking partial traces.) Therefore, each $\rho_{A,i}$ is block-diagonal in any $\hat N_A$-basis and $ \sum_{N_A} \Pi_{N_A}\rho_{A,i}\Pi_{N_A} = \rho_{A,i}$. Hence $\NA \rho = \sum_{N_A} \Pi_{N_A}\rho\Pi_{N_A} = \rho$, so $\NE\rho = 0$.
\end{proof}
We can gain further insight into symmetric separability by defining appropriate twirling operations. Let $\hat N = \{\hat n_j\}_{j=1,...,M}$ be a set of \emph{commuting} Hermitean operators on $\mathscr{H}_A \otimes \mathscr{H}_B$, and define the twirl
\beq
\mathcal{G}_{\hat N}[\rho]\coloneqq \int_{\mathcal{T}^M} dg\; T (g)\rho \, T^{\dagger}(g),
\eeq
where $T(g)=e^{i\sum g_j \hat n_j}=\prod e^{ig_j \hat n_j}$, and $dg$ is the Haar measure over the $M$-dimensional torus $\mathcal{T}^M$. (We will simply write $\mathcal{G}$ when there is no ambiguity.) It is immediate from translation-invariance of the Haar measure that $T(g)\mathcal{G}[\rho] T^{\dagger}(g)=\int_{\mathcal{T}^M} dg'\; T (g'+g)\rho \, T^{\dagger}(g'+g)=\mathcal{G}[\rho]$. Thus it follows
\beq\label{E:G_sym}
\forall g\in\mathcal{T}^M \colon [\mathcal{G}[\rho],T(g)]=0 \hspace{1cm} , \hspace{1cm} \forall \hat n_j \in \hat N \colon  [\mathcal{G}[\rho],\hat n_j]=0.
\eeq
\begin{proposition}\label{T:G_props}
The twirl is the restriction to $\mathscr{D}$ of a linear operator on $L(\mathscr{H}_A \otimes \mathscr{H}_B)$. It has the following properties:
\begin{enumerate}[leftmargin=1cm]
\item  It sends densities to $\hat N$-symmetric densities : $\mathcal{G}_{\hat N}[\mathscr{D}]\subset\dn = 
\{\rho\in\mathscr{D}\colon [\rho,\hat n]=0\text{ for all } \hat n \in \hat N\}$.

\item $\dn$ is the set of fixed points of $\mathcal{G}_{\hat N}$ : $\rho=\mathcal{G}_{\hat N}[\rho]$ iff $\rho\in\dn$. 

\item $\mathcal{G}_{\hat N}[\mathscr{D}] = \dn$.

\item $\mathcal{G}_{\hat N}$ is idempotent, $\mathcal{G}_{\hat N}^2=\mathcal{G}_{\hat N}$. 

\item $\mathcal{G}_{\hat N}\colon \mathscr{D}\to\dn$ is a retraction. 

\item $\mathcal{G}_{\hat N} = \mathcal{G}_{\hat n_M} \circ \dots \circ \mathcal{G}_{\hat n_1}$, where the ordering of the (commuting) $\hat n_j \in \hat N$ is immaterial.
\end{enumerate}
All these properties can be transposed to the extension of $\mathcal{G}_{\hat N}$ to $L(\mathscr{H}_A \otimes \mathscr{H}_B)$. In particular, $\mathcal{G}_{\hat N}$ is the projector onto the $\hat N$-symmetric subspace $L(\mathscr{H}_A \otimes \mathscr{H}_B)_{\hat N}$. 
\end{proposition}
\begin{proof}
(1) Clear from~\eqref{E:G_sym}. (2) If $\rho\in\dn$, then $[\rho,\hat n]=0$ for all $\hat n\in \hat N$, giving $[\rho,T(g)]=0$ for all $g\in \mathcal{T}^M$, hence $\mathcal{G}_{\hat N}[\rho]=\rho$. Conversely, if $\mathcal{G}_{\hat N}[\rho]=\rho$ then $\rho\in\dn$, from (1). The remaining properties are immediate consequences. 
\end{proof}
\begin{corollary}
Each $\mathcal{G}_{\hat n_j}$ corresponds to the nonselective measurement $(\cdot)|_{\hat n_j} = \sum_{n_j}\Pi_{n_j}(\cdot)\Pi_{n_j}$, as defined in the main text.
\end{corollary}
\begin{proof}
Both $\mathcal{G}_{\hat n_j}$ and $(\cdot)|_{\hat n_j}$ are the (unique) retraction onto $\mathscr{D}_{\hat n_j}$.
\end{proof}
Three cases will be especially relevant :
\begin{enumerate}[leftmargin=1cm]
	\item $\hat N=\hat N_{\times}=\{\hat N_A  \otimes \hat N_B\}$ on the one-dimensional torus $\mathcal{T}^1 = S^1$. Then $\mathcal{G}\colon\mathscr{D}\to\mathscr{D}_{\hat N_{\times}}$ is a retraction onto the  densities with conserved charge $N_A N_B$. 
	\item $\hat N=\hat N_{+}=\{\hat N_A \otimes 1_B + 1_A \otimes \hat N_B\}$ on the one-dimensional torus $\mathcal{T}^1 = S^1$. Then $\mathcal{G}\colon\mathscr{D}\to\mathscr{D}_{\hat N_{+}}$ is a retraction onto the densities with conserved charge $N_A +N_B$. 
	\item $\hat N=\hat N_{\textup{local}}=\{\hat N_A \otimes 1_B , 1_A \otimes \hat N_B\}$ on the two-dimensional torus $\mathcal{T}^2$. Then $\mathcal{G}\colon\mathscr{D}\to\mathscr{D}_{\hat N_{\textup{local}}}$ is a retraction onto the densities with conserved charges $N_A$ and $N_B$.
\end{enumerate}
\begin{proposition}
$\NE\rho\equiv 0$ on $\dnloc$.
\end{proposition}
\begin{proof}
$\rho\in\dnloc$ implies $\NA \rho = \rho$.
\end{proof}
Inspection of the proof of Proposition~\ref{T:Ma} gives the next result.
\begin{proposition}\label{T:symsepn_symsepnloc}
If $\hat N = \hat N_{\times}$ or $\hat N = \hat N_{+}$, then $\symsepn = \symsepnloc$.
\end{proposition}
For $\hat N_{\textup{local}}$, the unitary representation is of separable form $T(g,g')=T_A(g) \otimes T_B (g')$ . For $\hat N_{+}$, it is still separable but with $g=g'$. For $\hat N_{\times}$, the unitaries are no longer separable. For $\hat N_{\textup{local}}$ and $\hat N_{+}$, and for $\xi\in \mathscr{H}_A$, let
\beq
\mathcal{G}_A[\xi]\coloneqq \int_{S^1} \frac{dg}{2\pi}\; T_A (g)\xi \, T_A^{\dagger}(g),
\eeq
where $T_A(g) = e^{ig \hat N_A}$. Define $\mathcal{G}_B[\xi]$ similarly. 
\begin{proposition}\label{T:inclusions}
Let $\mathcal{G} = \mathcal{G}_{\hat N_{\textup{local}}}$. Then 
\begin{enumerate}[leftmargin=1cm]
\item $\mathcal{G}[\sep]$ has nonzero measure in $\dnloc$.
\item  $\mathcal{G}[\sep] = \sepnloc$.
\item $\mathcal{G}[\sep] = \symsepnloc$. 
\end{enumerate}
If $\mathcal{G} = \mathcal{G}_{\hat N_{+}}$, then only properties (1) and (2) convey:
\begin{enumerate}[leftmargin=1cm]
\item $\mathcal{G}[\sep]$ has nonzero measure in $\mathscr{D}_{\hat N_{\textup{+}}}$.
\item  $\mathcal{G}[\sep] = \mathscr{D}^{\textup{sep}}_{\hat N_{\textup{+}}}$.
\end{enumerate}
If $\mathcal{G} = \mathcal{G}_{\hat N_{\times}}$, then only property (1) and half of (2) convey:
\begin{enumerate}[leftmargin=1cm]
\item $\mathcal{G}[\sep]$ has nonzero measure in $\mathscr{D}_{\hat N_{\times}}$.
\item[2*.]  $\mathcal{G}[\sep] \supset \mathscr{D}^{\textup{sep}}_{\hat N_{\times}}$.
\end{enumerate}
\end{proposition}
\begin{proof}
We prove the properties for $\mathcal{G}=\mathcal{G}_{\hat N_{\textup{local}}}$.

(1) Let $\textup{int\,}\sep$ be the (nonempty) interior of $\sep$~\cite{Wen2023separable}, and let $\rho\in\mathcal{G}[\textup{int\,}\sep]$ so that $\tilde\rho\in\mathcal{G}^{-1}[\rho]$ is in the (nonempty) interior of $\sep$, and is surrounded by an open ball of separable states. Consider any density $\sigma\in\dnloc$. For $\epsilon>0$ small enough, $(1-\epsilon)\tilde\rho + \epsilon\mathcal{G}^{-1}[\sigma]\in\textup{int\,}\sep$, i.e.~$(1-\epsilon)\tilde\rho + \epsilon\tilde\sigma$ is interior separable for any $\tilde\sigma\in\mathcal{G}^{-1}[\sigma]$. Thus, $(1-\epsilon)\rho +\epsilon\sigma = \mathcal{G}[(1-\epsilon)\tilde\rho+ \epsilon\mathcal{G}^{-1}[\sigma]]\in\mathcal{G}[\textup{int\,}\sep]$, showing that $\mathcal{G}[\textup{int\,}\sep]$ contains an open ball of size $\epsilon\,\text{vol}\,\dnloc >0$. By $\mathcal{G}[\textup{int\,}\sep] \subset \mathcal{G}[\sep]$, we have our result.

(2) Because $\mathcal{G}$ fixes $\dnloc$,  we have $\sepnloc = \mathcal{G}[\sepnloc]\subset \mathcal{G}[\sep]$. Conversely, by linearity and because unitaries in the twirl are of the form  $T(g,g')=T_A(g) \otimes T_B (g')$, we have $\mathcal{G}[\sep]\subset \sepnloc$. 

(3) For any $\rho=\sum_i p_i \rho_{A,i}\otimes \rho_{B,i}\in\sep$, we have $\mathcal{G}[\rho]\in \dnloc$ by Proposition~\ref{T:G_props}, and $\mathcal{G}[\rho] = \sum_i p_i \mathcal{G}_A[\rho_{A,i}]\otimes \mathcal{G}_B[\rho_{B,i}]$. Because $[\mathcal{G}_A[\rho_{A,i}], T_A (g)]=0$ for all $g\in S^1$, we have $[\mathcal{G}_A[\rho_{A,i}], \hat N_{A}]=0$. A similar argument gives $[\mathcal{G}_B [\rho_{B,i}], \hat N_{B}]=0$, whence $[\mathcal{G}_A[\rho_{A,i}]\otimes \mathcal{G}_B[\rho_{B,i}], \hat N_{A} \otimes 1_B]=0=[\mathcal{G}_A[\rho_{A,i}]\otimes \mathcal{G}_B[\rho_{B,i}], 1_A \otimes \hat N_{B}]$. We conclude that $\mathcal{G}[\rho]\in\symsepnloc$, whence $\mathcal{G}[\sep]\subset\symsepnloc$.  Conversely, because $\mathcal{G}$ fixes $\dnloc$, $\symsepnloc = \mathcal{G}[\symsepnloc] \subset \mathcal{G}[\sep]$.

If $\mathcal{G}=\mathcal{G}_{\hat N_{+}}$, the proof of (3) fails because $g=g'$ in the unitaries of the twirl. If $\mathcal{G}=\mathcal{G}_{\hat N_{\times}}$, the proofs of (2) and (3) fail because unitaries are not of separable form.
\end{proof}

\begin{proposition}\label{T:symsepn_in_dnloc}
If $\hat N = \hat N_{+}$ or $\hat N = \hat N_{\times}$, then $\symsepn = \sepnloc$. This set has nonzero measure in $\dnloc$.
\end{proposition}
\begin{proof}
Immediate from Propositions~\ref{T:symsepn_symsepnloc}, and~\ref{T:inclusions}.
\end{proof}
The Symmetric Separability Problem is thus reduced to separability within a restricted class of locally symmetric states. The NE is no longer a witness of symmetric inseparability within this class, because it vanishes identically on it. However, the reduction will enable us to apply, to some extent, known separability results to the problem of symmetric separability.
To that end, let us consider the space of $\hat N_{\textup{local}}$-symmetric linear operators on $\mathscr{H}_A \otimes \mathscr{H}_B$, that is $L(\mathscr{H}_A \otimes \mathscr{H}_B)_{\hat N_{\textup{local}}}$. Because $\hat N_A\otimes 1_B$ and $1_A \otimes \hat N_B$ commute, they have a common diagonalizing basis $\{\Ket{\phi_j^{p,q}}\}$, and 
\beq\label{E:charge_sectors_appendix}
L(\mathscr{H}_A \otimes \mathscr{H}_B)_{\hat N_{\textup{local}}} = \bigoplus_{p,q} K_{p,q}
\eeq
where $K_{p,q}$ is the eigenspace of operators of charge $p,q$. There is no other restriction within each block, so operators in $K_{p,q}$ are of the form $\sum_{jk} m_{jk} \Ket{\phi_j^{p,q}} \Bra{\phi_k^{p,q}}$ for any matrix $\{m_{jk}\}$. In other words, $K_{p,q}=L(V_{p,q})$, with $V_{p,q}=\textup{span}\{\Ket{\phi_j^{p,q}}\}$. 
\begin{proposition}
For some integer $M$, there is an orthonormal basis $\{\Ket{u_j}\otimes\Ket{v_k}\}_{1\leq j,k \leq M}$ of $V_{p,q}$ with $\{\Ket{u_j}\}\subset \mathscr{H}_A$ and $\{\Ket{v_k}\}\subset \mathscr{H}_B$.
\end{proposition}
\begin{proof}
We proceed inductively. Begin with a Schmidt decomposition of $\Ket{\phi_1^{p,q}}$,
\beq
\Ket{\phi_1^{p,q}} = \sum_{i=1}^r \alpha_i \Ket{u_i}\otimes \Ket{v_i},
\eeq
with orthonormal sets $\{\Ket{u_1}, \dots , \Ket{u_r}\}\subset \mathscr{H}_A$, and $\{\Ket{v_1}, \dots , \Ket{v_r}\}\subset \mathscr{H}_B$. For $\Ket{\phi_2^{p,q}}$, let
\beq
\Ket{\phi_2^{p,q}} = \eta + \xi,
\eeq
with $\eta\in\textup{span}\{\Ket{u_j}\otimes \Ket{v_k}\}_{1\leq j,k \leq r}$, and $\xi$ in the orthogonal complement of $\textup{span}\{\Ket{u_j}, \Ket{v_k}\}_{1\leq j,k \leq r}$. Schmidt decompose $\xi$ as
\beq
\xi = \sum_{i=r+1}^{r+s} \alpha_i \Ket{u_i}\otimes \Ket{v_i},
\eeq
so that $\{\Ket{u_1}, \dots , \Ket{u_{r+s}}\}$ and $\{\Ket{v_1}, \dots , \Ket{v_{r+s}}\}$ are orthonormal in $\mathscr{H}_A$ and $\mathscr{H}_B$, respectively, and $\Ket{\phi_1^{p,q}}$, $\Ket{\phi_2^{p,q}}$ are in the span of $\{\Ket{u_j}\otimes\Ket{v_k}\}_{1\leq j,k \leq r+s}$. The process can be continued until each $\Ket{\phi_j^{p,q}}$ is decomposed in such manner, and stops after finitely many steps. The resulting set $\{\Ket{u_j}\otimes\Ket{v_k}\}_{1\leq j,k \leq M}$ is as in the statement.
\end{proof}
The subset of densities in $K_{p,q}$ is, in our notation, the set $\mathscr{D}(V_{p,q}) = \mathscr{D}(h_A^{p,q} \otimes h_B^{p,q})$, where $h_A^{p,q} = \textup{span}\{\Ket{u_j}\}\subset \mathscr{H}_A$, and $h_B^{p,q} = \textup{span}\{\Ket{v_k}\}\subset \mathscr{H}_B$. The subset of densities in $L(\mathscr{H}_A \otimes \mathscr{H}_B)_{\hat N_{\textup{local}}}$ is obtained as the convex hull
\beq
\dnloc = \textup{conv}\bigcup_{p,q}\mathscr{D}(h_A^{p,q} \otimes h_B^{p,q}).
\eeq
Note that $\textup{conv}\bigcup_{p,q}\sep(h_A^{p,q} \otimes h_B^{p,q})\subset\sepnloc$, and \emph{a priori} the inclusion could be proper, because a combination of an entangled state and a separable state, or even a combination of two entangled states, can be separable~\cite{Steiner2003generalized}. But since the charge sectors correspond to disjoint matrix blocks, separation by mixing cannot occur by combining elements of non-equal charges. Summing up, we obtain
\begin{proposition}\label{T:conv_hulls}
$\sepnloc = \textup{conv}\bigcup_{p,q}\sep(h_A^{p,q} \otimes h_B^{p,q})$, where $h_A^{p,q} \subset \mathscr{H}_A$, $h_B^{p,q} \subset \mathscr{H}_B$, and $p,q$ is the charge index associated with $\hat N_{\textup{local}}$-symmetry. Elements from different sectors are linearly independent.
\end{proposition}
We are now in the position to apply known separability results to the different charge sectors. In particular, numerical evidence~\cite{Zyczkowski1998volume1, Zyczkowski1999volume2} suggests that each ratio 
\beq
\frac{\textup{vol}\,\sep(h_A^{p,q} \otimes h_B^{p,q}) }{ \textup{vol}\,\mathscr{D}(h_A^{p,q} \otimes h_B^{p,q})}
\eeq
 is exponentially small in the dimension $\textup{dim}\, h_A^{p,q} \otimes h_B^{p,q}$.\footnote{Entangled states are known to be vastly more numerous than separable ones. The volume of the separable space $\sep_{1_{AB}}$ is not yet known analytically, but is known to be nonzero with asymptotic lower bounds $\exp [-(d_A d_B)^{\gamma} \ln d_A d_B]$ in common measures, where $\gamma = 1$~\cite{Vidal1999robustness} or $\gamma =2$~\cite{Szarek2005volume}. Nontrivial upper bounds for multipartite separability are also known~\cite{Szarek2005volume}. For one family of natural measures, numerical evidence suggests that the volume of $\sep_{1_{AB}}$ decreases exponentially in Hilbert-space dimension, that is $\mu(\sep_{1_{AB}})= k_1 e^{-k_2 d_A d_B}$ for some positive constants $k_1 , k_2$~\cite{Zyczkowski1998volume1, Zyczkowski1999volume2}.} Geometrically, each $\sep (h_A^{p,q} \otimes h_B^{p,q})$ is comprised between two hyperplanes $C_{p,q}$ and $F_{p,q}$~\cite{Lockhart2002preserving}. The sandwiched hyperslab containing all separable states of sector $p,q$ (and many entangled states too) has thickness going to zero as the dimensions of $h_A^{p,q}$ and $h_B^{p,q}$ tend to infinity. In that limit, the separable states of sector $p,q$ cluster near a hyperplane which contains the maximally-mixed state.

The convex hull of small sets (or even nullsets) need not be small. To conclude from Proposition~\ref{T:conv_hulls} and the smallness of sectors that $\sepnloc$ itself is small, we need to remember that these sectors are mutually orthogonal. We will rely on the following geometrical intuition (see Figure~\ref{F:hull_ortho}) : For a subset $W\subset \mathbb{R}^n$, and its orthogonal complement $W^{\perp}$, let both $C_1\subset W$, $C_2\subset W^{\perp}$ be compact convex, and let $S_2\subset C_1$, $S_2 \subset C_2$ be convex subsets. If both $\textup{vol}\, S_1 / \textup{vol}\, C_1$ and $\textup{vol}\, S_2/\textup{vol}\, C_2$ tend to zero exponentially, so does 
\beq
\frac{\textup{vol conv} (S_1\cup S_2)}{\textup{vol conv} (C_1\cup C_2)}.
\eeq
Applying this intuition to the problem of the size of $\sepnloc$, we have the following:
\begin{conjecture}
If all charge sectors of $\dnloc = \textup{conv}\bigcup_{p,q}\mathscr{D}(h_A^{p,q} \otimes h_B^{p,q})$ have large dimension, we expect $\sepnloc$, the population of (symmetrically) separable states in $\dnloc$, to be exponentially small (in some combination of sector dimensions). 
\end{conjecture}
If $\hat N = \hat N_{+}$ or $\hat N = \hat N_{\times}$, we need to replace the sectors $K_{p,q}$ of Eqn.~\eqref{E:charge_sectors_appendix} by sectors of the form $K_n$ with $n=p+q$ or $n=pq$, respectively. Then the argument follows with minor modifications, giving:
\begin{conjecture}
If all charge sectors of $\dn = \textup{conv}\bigcup_{n}\mathscr{D}(h_A^{n} \otimes h_B^{n})$ have large dimension, we expect $\sepn$, the population of separable states in $\dn$, to be exponentially small (in some combination of sector dimensions). 
\end{conjecture}
\begin{figure}
\includegraphics{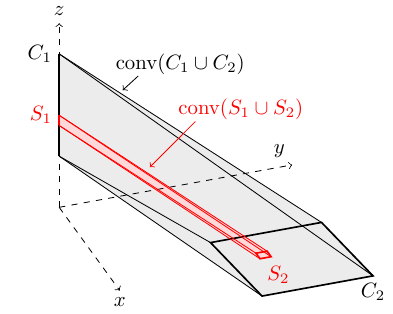}
\caption{\textbf{Convex hull of two compact convex sets belonging to mutually orthogonal spaces.} $C_1\subset W$ and $C_2\subset W^{\perp}$ are compact convex sets in $W$ and $W^{\perp}$, respectively, two mutually orthogonal subspaces of $\mathbb{R}^n$. (In the figure, $W=\textup{span}\{z\}$ and $W^{\perp}=\textup{span}\{x,y\}$.) $S_1$ and $S_2$ are convex subsets of $C_1$ and $C_2$, respectively. If both $\textup{vol}\, S_1 / \textup{vol}\, C_1$ and $\textup{vol}\, S_2/\textup{vol}\, C_2$ tend to zero, so does 
$\textup{vol conv} (S_1\cup S_2) / \textup{vol conv} (C_1\cup C_2)$. In general, it is not enough that only one among $\textup{vol}\,S_1, \textup{vol}\,S_2$ be small.}\label{F:hull_ortho}
\end{figure}

\subsection{Background : Concentration of EE}
\noindent The von Neumann entropy of a state $\rho$ (pure or mixed) is 
\beq\label{E:vN}
S(\rho) = -\text{Tr}(\rho\log\rho).
\eeq
We consider bipartite states in $\mathscr{H}_A \otimes \mathscr{H}_B$, where  $\mathscr{H}_A$ and  $\mathscr{H}_B$ have respective dimensions $d_A$ and $d_B$. The density corresponding to a pure state $\ket\psi$ will be written $\psi = \Ket\psi\Bra\psi$. The partial trace of a state (pure or mixed) will be denoted $\rho_A = \text{Tr}_B \rho$. The bipartite entanglement entropy (EE) of $\rho$ on $A$ is 
\beq
S(\rho_A) = S(\text{Tr}_B \rho),
\eeq
where $S$ is the von Neumann entropy~\eqref{E:vN}. For pure states, $S(\psi_A)=S(\psi_B)$.  

\begin{lemma}[Hayden \emph{et al}.]\label{T:hayden}
If $d_A \geq 3$, the bipartite EE on pure states is Lipshitz with constant $\eta = \sqrt{8}\log d_A$ with respect to the Euclidean norm $\lVert\cdot\rVert_2$, i.e.
\beq
|S(\psi_A) - S(\phi_A) | \leq \eta \lVert \Ket\psi - \Ket\phi \rVert_2 
\eeq
for any pure states $\Ket\psi, \Ket\phi$.
\end{lemma}
Lemma~\ref{T:hayden} and the fact that the (projective) space of pure states $P(\mathbb{C}^{d_A d_B}) \simeq S^{2d_A d_B-1}/U(1)$ is essentially a sphere implies the phenomenon of \emph{concentration of measure} for the pure-state bipartite EE. The following lemma is the original example of measure concentration~\cite{Ledoux2001concentration}, and will be applied to different measures of entanglement.
\begin{lemma}[Lévy's Lemma]\label{T:levy}
Let $f : S^{n}\to \mathbb{R}$ be a Lipshitz function with constant $\eta$, i.e.
\beq
|f(x)-f(y)|\leq \eta\lVert x-y\rVert_2, \hspace{2cm}\forall x,y \in S^{n},
\eeq
and let $X$ be a random variable on $S^n$ equipped with the Haar measure. Then  
\beq
\textup{Prob}(|f(X)-\mathbb{E}f(X)| > \alpha) \leq 2e^{-c(n+1)\alpha^2 / \eta^2},
\eeq
where $\mathbb{E}f(X)$ is the expectation value of $f(X)$, and $c$ is a positive constant that may be chosen as $c=(18\pi^3)^{-1}$. 
\end{lemma}
The lemma implies that a Lipshitz function on a sphere is essentially constant and equal to its average. By Lemmas~\ref{T:hayden},~\ref{T:levy} we thus have
\begin{proposition}[Hayden \emph{et al.}]\label{T:hayden_EE}
For a pure state $\Ket\phi \in \mathscr{H}_A \otimes \mathscr{H}_B$ chosen uniformly at random,
\beq\label{E:concentration_EE}
\textup{Prob}(|S(\phi_A) - \mathbb{E}S(\phi_A)| > \alpha) \leq 2e^{-c d_A d_B \alpha^2 / \eta^2},
\eeq
where $\mathbb{E}S(\phi_A)$ is the expectation value of the EE on pure states equipped with the Haar measure, and $c$ is a positive constant that may be chosen as $c=2(18\pi^3)^{-1}$. 
\end{proposition}
Hence pure-state EE is essentially constant and equal to its average. Together with the provable fact that 
\beq
\mathbb{E}S(\phi_A) > \log d_A - \frac{1}{2\ln2}\frac{d_A}{d_B},
\eeq
one concludes that as long as $d_A \ll d_B$, the EE is essentially always close to its maximum of $\log d_A$. Pure bipartite states generically have near-maximal entanglement~\cite{Hayden2006aspects}. 

\subsection{Concentration of NE}\label{S:conc_NE_appendix}
Let $\rho$ be a density of $\mathscr{H}_A \otimes \mathscr{H}_B$, and let $\hat N_A$ be an observable on $A$. The nonselective measurement of $\hat N_A$ results in the state $\NA \rho = \sum_{N_A} \Pi_{N_A}\rho\Pi_{N_A}$, where $\Pi_{N_A}$ is the projector onto the sector of charge $N_A$. The \emph{number entanglement} (NE)~\cite{Ma2022symmetric} of $\rho$ with respect to $\hat N_A$ is
\beq
\NE\rho = S(\NA\rho) - S(\rho).
\eeq
Note that on pure states, $\NE\phi = S(\NA\phi)$. For densities $\rho$ on $\mathscr{H}_A \otimes \mathscr{H}_B$ we will write $\Ket{\pu\rho}$ for a purification in a common larger space $\mathscr{H}_A \otimes \mathscr{H}_B \otimes \mathscr{H}_C$, i.e.
\beq
\text{Tr}_C\pu\rho = \text{Tr}_C\Ket{\pu\rho}\Bra{\pu\rho} = \rho.
\eeq
Using purifications of $\rho$ and its post-measurement relative $\NA\rho$, the NE can expressed as a difference of pure state EE's:
\beq\label{E:NE_as_EE}
\begin{aligned}
\NE\rho &= S(\NA\rho) - S(\rho)\\
	&= S(\text{Tr}_C\pu{\NA\rho})-S(\text{Tr}_C\pu{\rho})\\
	&= S\big((\pu{\NA\rho})_{A\cup B}\big)- S\big((\pu{\rho})_{A\cup B}\big).
\end{aligned}
\eeq
\begin{proposition}\label{T:NE_lipshitz}
When $\mathscr{H}_A\otimes\mathscr{H}_B$ has dimension $d_A d_B\geq 3$, the NE with respect to $\hat N_A$ is Lipshitz on any space of purifications, i.e.
\beq
|\NE\sigma - \NE\rho|\leq \eta \big\lVert \Ket{\pu{\sigma}}- \Ket{\pu{\rho}}\big\rVert_2,
\eeq
with $\eta = 4\sqrt{2}\log (d_A d_B)$, and where $\Ket{\pu{\sigma}}, \Ket{\pu{\rho}}$ are any purifications of $\sigma, \rho$ on a common space. Obviously, $\eta$ may be relaxed to $4\sqrt{2}\log D$, where $D \geq d_A d_B$ is the dimension of the purification space. 
\end{proposition}
\begin{proof}
From~\eqref{E:NE_as_EE}, we have
\beq
\begin{aligned}
|\NE\sigma - \NE\rho| \leq &\;|S\big((\pu{\NA\sigma})_{A\cup B}\big)- S\big((\pu{\NA\rho})_{A\cup B}\big)|\\
	 &+ |S\big((\pu{\sigma})_{A\cup B}\big)- S\big((\pu{\rho})_{A\cup B}\big)|
\end{aligned}
\eeq
From Lemma~\ref{T:hayden}, it follows that
\beq\label{E:ineq}
|\NE\sigma - \NE\rho| \leq \eta_0\left(\big\lVert \Ket{\pu{\NA\sigma}}- \Ket{\pu{\NA\rho}}\big\rVert_2
+ \big\lVert \Ket{\pu{\sigma}}- \Ket{\pu{\rho}}\big\rVert_2\right),
\eeq
where $\eta_0 = \sqrt{8}\log (d_A d_B)$. Note that inequality~\eqref{E:ineq} is true for any purifications of $\sigma$, $\NA\sigma$, $\rho$, and $\NA\rho$. Fixing $\pu{\NA\rho}$ and $\pu{\rho}$, we thus have
\beq\label{E:ineq_inf}
\begin{aligned}
|\NE\sigma - \NE\rho| \leq &\hspace{5pt}\eta_0\inf_{\Ket{\pu{\NA\sigma}}}\big\lVert \Ket{\pu{\NA\sigma}}- \Ket{\pu{\NA\rho}}\big\rVert_2\\
	&+ \eta_0\inf_{\Ket{\pu{\sigma}}}\big\lVert \Ket{\pu{\sigma}}- \Ket{\pu{\rho}}\big\rVert_2,
\end{aligned}
\eeq
where the infima are taken on purifications of $\NA\sigma$ and $\sigma$. Now
\beq
\begin{aligned}
\inf_{\Ket{\pu{\NA\sigma}}}\big\lVert \Ket{\pu{\NA\sigma}}- \Ket{\pu{\NA\rho}}\big\rVert_2^2 &=
2\inf_{\Ket{\pu{\NA\sigma}}}\left(1-\text{Re}\langle\pu{\NA\sigma}\Ket{\pu{\NA\rho}}\right)\\
	&= 2\Bigg(1-\sup_{\Ket{\pu{\NA\sigma}}}\langle\pu{\NA\sigma}\Ket{\pu{\NA\rho}}\Bigg).
\end{aligned}
\eeq
Note that $\sup_{\Ket \psi}\text{Re}\langle\psi\Ket\phi = \sup_{\Ket \psi}\langle\psi\Ket\phi$ whenever one is allowed to change the global phase of $\Ket\psi$. Using Uhlmann's theorem, we find
\beq
\inf_{\Ket{\pu{\NA\sigma}}}\big\lVert \Ket{\pu{\NA\sigma}}- \Ket{\pu{\NA\rho}}\big\rVert_2^2 = 
2\left(1-\sqrt{F(\NA\sigma , \NA\rho)}\right),
\eeq
where $F(\cdot , \cdot)$ is the fidelity between two quantum states. Because the quantum channel $\mathcal{E}(\cdot) = \NA{(\cdot)}$ is trace-preserving and completely positive, monotonicity of fidelity implies
\beq
F(\NA\sigma , \NA\rho ) \geq F(\sigma, \rho).
\eeq
Hence,
\beq\label{E:Bures}
\inf_{\Ket{\pu{\NA\sigma}}}\big\lVert \Ket{\pu{\NA\sigma}}- \Ket{\pu{\NA\rho}}\big\rVert_2^2 \leq 2\left(1-\sqrt{ F(\sigma, \rho)}\right) =\inf_{\Ket{\pu{\sigma}}}\big\lVert \Ket{\pu{\sigma}}- \Ket{\pu{\rho}}\big\rVert_2^2.
\eeq
Combining the above with~\eqref{E:ineq_inf} finally gives
\beq
|\NE\sigma - \NE\rho|\leq \eta \big\lVert \Ket{\pu{\sigma}}- \Ket{\pu{\rho}}\big\rVert_2,
\eeq
with $\eta = 4\sqrt{2}\log (d_A d_B)$. Note that $\Ket{\pu{\sigma}}, \Ket{\pu{\rho}}$ are any purifications of $\sigma, \rho$ on a common space.
\end{proof}
\begin{corollary}\label{T:Bures}
When $\mathscr{H}_A\otimes\mathscr{H}_B$ has dimension $d_A d_B\geq 3$, the NE with respect to $\hat N_A$ is (Lipshitz) continuous in the Bures distance,
\beq
|\NE\sigma - \NE\rho|\leq \eta D_B (\sigma , \rho).
\eeq
\end{corollary}
\begin{proof}
Immediate from the previous proposition, Eqn.~\eqref{E:Bures}, and the definition of the Bures distance,
\beq
D_B(\sigma , \rho) = \sqrt{2\left(1-\sqrt{ F(\sigma, \rho)}\right)}.
\eeq
We note that the Bures distance defines a metric on $\mathscr{D}$~\cite{Hayashi2006quantum, Bhatia2019bures}.\end{proof}
Normalized pure state $\Ket\Psi \in \mathscr{H}_A \otimes \mathscr{H}_B \otimes \mathscr{H}_C$ live on the sphere $S^{2d_A d_B d_C -1}$, and are many-to-one surjective onto the densities $\rho\in \mathscr{D}(\mathscr{H}_A \otimes \mathscr{H}_B)$ by the action of the partial trace :
\beq
\begin{aligned}
\pi \colon S^{2d_A d_B d_C -1} &\to \mathscr{D}(\mathscr{H}_A \otimes \mathscr{H}_B)\\
	\Ket\Psi \phantom{bla}		&\mapsto \phantom{blab}\text{Tr}_C \Psi.
\end{aligned}
\eeq
The sphere is thus partitioned into fibers
\beq
\pi^{-1}(\rho) = \big\{(\mathcal{W}\otimes \mathcal{U})\Ket{\pu\rho} \colon \mathcal{W}\in C(\rho),\; \mathcal{U}\in U(d_C)\big\},
\eeq
where $\Ket{\pu\rho}\in \mathscr{H}_A \otimes \mathscr{H}_B \otimes \mathscr{H}_C$ is any fixed purification of $\rho$, and $C(\rho) = \{\mathcal{W}\in U(d_A d_B) \colon \mathcal{W}\rho\mathcal{W}^{\dagger} = \rho\}$. From now on, we will simply write $\mathscr{D}$ instead of $\mathscr{D}(\mathscr{H}_A \otimes \mathscr{H}_B)$, for brevity. The convex set of densities inherits a measure from that construction given as 
\beq\label{E:muBvsmuHaar}
\mu_{\mathscr{D}} (b) = \mu_{\text{Haar}}(\cup_{\rho\in b} \;\pi^{-1}(\rho)),
\eeq
where $\mu_{\text{Haar}}$ is the normalized Haar measure on the sphere. Clearly, the inherited measure is a probability measure on $\mathscr{D}$ : for $b\subseteq\mathscr{D}$, we will use $\text{Prob}_{\mathscr{D}}(b)$ and $\mu_{\mathscr{D}} (b)$ interchangeably. States chosen at random with respect to $\mu_{\mathscr{D}}$ agree with the `rank-$s$ random states' of~\cite{Hayden2006aspects}, and with the `random induced mixed states' of~\cite{Shapourian2021entanglement}. They are based on the old practice of inducing probability measures on mixed states by partial tracing~\cite{Braunstein1996geometry, Hall1998random, Zyczkowski2001induced}. Now let
\beq
\begin{aligned}
f = \Delta S_{\hat N_A}\circ \pi\hspace{5pt}\colon\; S^{2d_A d_B d_C -1} &\to \mathbb{R}\\
	\Ket\Psi \phantom{bla}		&\mapsto \NE{\text{Tr}_C \Psi}.
\end{aligned}
\eeq
By Proposition~\ref{T:NE_lipshitz}, $f$ is Lipshitz on the sphere,
\beq
\big| f(\Ket\Psi) - f(\Ket\Phi)\big|\leq \eta\big\lVert\Ket\Psi - \Ket\Phi \big\rVert_2,
\eeq
with constant $\eta = 4\sqrt{2}\log (d_A d_B)$. By Lévy's Lemma~\ref{T:levy}, it presents the phenomenon of concentration of measure, and is thus essentially always equal to its expectation value :
\beq\label{E:conc_f}
\text{Prob}\big(\big| f(\Ket\Psi) - \mathbb{E}f(\Ket\Psi) \big| > \alpha \big) \leq 2e^{-cd_A d_B d_C \alpha^2 / \eta^2},
\eeq
where $\Ket\Psi$ is chosen at random on $S^{2d_A d_B d_C -1}$ equipped with the Haar measure. The constant $c$ is as in Proposition~\ref{T:hayden_EE}. Moreover, the expectation $\mathbb{E}\NE\rho$ on $\mathscr{D}$ is equal to the expectation $\mathbb{E}f(\Psi)$ on $S^{2d_A d_B d_C -1}$. Indeed, by the definition of the Lebesgue integral,
\beq\label{E:equal_expectations}
\begin{aligned}
\mathbb{E}f(\Ket\Psi) = \int_{\text{sphere}} f(\Ket\Psi) \;d\mu_{\text{Haar}} &= \int_0^{\infty}\mu_{\text{Haar}}(\{\Ket\Psi \colon f(\Ket\Psi)>x\})\; dx\\
	&= \int_0^{\infty}\mu_{\mathscr{D}}(\{\rho \colon \NE\rho >x\})\; dx\\
	&= \int_{\mathscr{D}} \NE\rho \; d\mu_{\mathscr{D}}\\
	&=\mathbb{E}\NE\rho ,
\end{aligned}
\eeq
where we have used~\eqref{E:muBvsmuHaar}. Eqns.~\eqref{E:muBvsmuHaar} and~\eqref{E:equal_expectations} allows us to rewrite~\eqref{E:conc_f} as
\beq\label{E:conc_NE}
\text{Prob}_{\mathscr{D}}(| \NE\rho - \mathbb{E}\NE\rho | > \alpha) \leq 2e^{-cd_A d_B d_C \alpha^2 / \eta^2}.
\eeq
In this inequality, $\rho$ is chosen at random with respect to the measure $\mu_{\mathscr{D}}$, which by construction depends on $d_C$. (This measure would be called $P_{N,K}$ with $N=d_A d_B$ and $K=d_C$ in the terminology of~\cite{Zyczkowski2001induced}.) Taking $(\mathscr{H}_A \otimes \mathscr{H}_B)^{\otimes 2}$ as the canonical purification space (i.e. the measure $P_{N,N}$ or equivalently the \emph{Hilbert-Schmidt measure} also defined in~\cite{Zyczkowski2001induced}, which also shows a relation to the complex Ginibre ensemble), we find:
\begin{proposition}\label{T:conc_NE_canonical}
With overwhelming probability, a (mixed or pure) random state's NE is almost exactly equal to the average NE on mixed states:
\beq\label{E:conc_NE_canonical}
\textup{Prob}_{\mathscr{D}}(| \NE\rho - \mathbb{E}\NE\rho | > \alpha) \leq 2e^{-cd_A^2 d_B^2 \alpha^2 / \eta^2},
\eeq
where $\rho$ is chosen at random with respect to the measure $\mu_{\mathscr{D}}$ inherited from the canonical construction of purifications.
\end{proposition}
A refinement of Proposition~\ref{T:conc_NE_canonical} which will be useful in the next section is obtained by considering only states that lie in a closed convex subset $\mathscr{K}\subseteq\mathscr{D}$. Consider the sequence of smooth mappings :
\beq\label{E:maps_1}
\begin{aligned}
&P(\mathbb{C}^{d_A^2 d_B^2})\simeq S^{2d_A^2 d_B^2 - 1}/U(1) &&\overset{i}{\longrightarrow} &&(\mathscr{H}_A \otimes \mathscr{H}_B)^{\otimes 2}\big\lvert_{\text{ran } i}
 &&\overset{j}{\longrightarrow}&&\mathscr{D}\left((\mathscr{H}_A \otimes \mathscr{H}_B)^{\otimes 2}\right)
 &&\overset{\text{Tr}_{\mathscr{H}_A \otimes \mathscr{H}_B}}{\longrightarrow}&&\mathscr{D}\\
 &\hspace{5em}\Ket{\Psi} &&\longmapsto &&\phantom{llllll}\Ket{\Psi} &&\longmapsto &&\phantom{lllllllll}\Ket{\Psi}\Bra{\Psi} &&\hspace{11pt}\longmapsto &&\rho
\end{aligned}
\eeq
where $(\cdot)\big\lvert_{\text{ran } i}$ means restriction to the range of $i$. On the one hand, $j\circ i (P(\mathbb{C}^{d_A^2 d_B^2}))$ is the boundary of the closed convex set $\mathscr{D}\left((\mathscr{H}_A \otimes \mathscr{H}_B)^{\otimes 2}\right)$, whose interior is the convex open set of impure densities on $(\mathscr{H}_A \otimes \mathscr{H}_B)^{\otimes 2}$. On the other hand, $\text{Tr}_{\mathscr{H}_A \otimes \mathscr{H}_B}^{-1}(\mathscr{K})$ is a closed convex subset of $\mathscr{D}\left((\mathscr{H}_A \otimes \mathscr{H}_B)^{\otimes 2}\right)$, by the linearity of $\text{Tr}_{\mathscr{H}_A \otimes \mathscr{H}_B}$. Therefore, the intersection of these sets,
\beq
\mathfrak{X}=j\circ i (P(\mathbb{C}^{d_A^2 d_B^2})) \cap \text{Tr}_{\mathscr{H}_A \otimes \mathscr{H}_B}^{-1}(\mathscr{K})
\eeq
(or equivalently the intersection of their boundaries) is closed, compact, and is the tangential intersection of two real topological manifolds of respective dimension $2d_A^2 d_B^2 -1$ and $d < 2d_A^4 d_B^4$. The same will hold for the homeomorphic preimage of $\mathfrak{X}$ in $(\mathscr{H}_A \otimes \mathscr{H}_B)^{\otimes 2}\big\lvert_{\text{ran } i}$, the set $\chi=j^{-1}(\mathfrak{X})$.\footnote{We emphasize that $\mathfrak{X}$ is a subset of $j((\mathscr{H}_A \otimes \mathscr{H}_B)^{\otimes 2}\big\lvert_{\text{ran } i})$, and that $j^{-1}$ is a homeomorphism between this set and the set $(\mathscr{H}_A \otimes \mathscr{H}_B)^{\otimes 2}\big\lvert_{\text{ran } i}$.}  A non-transversal intersection of manifolds will not be a manifold in general. The local Euclidean dimension may change from point to point, and some points may even lack a Euclidean neighborhood if they present a bifurcation for instance. Let $n_{\Ket{\Psi}}\geq 1$ be the dimension of an open Euclidean neighborhood of $\Ket{\Psi}\in\chi$, and if $\Ket{\Psi}$ does not have a Euclidean neighborhood, set $n_{\Ket{\Psi}}=0$. Define $D=\text{max}_{\Ket{\Psi}\in \chi} n_{\Ket{\Psi}}$, and let $\chi_0$ be the closure of the largest topological manifold of dimension $D$ in $\chi$. The integer $D$ will be called the \emph{purification dimension} of the set $\mathscr{K}$, or somewhat abusively the \emph{dimension} of $\mathscr{K}$ when the context is clear, and $\chi_0$ will be called the \emph{purifying manifold} of $\mathscr{K}$. We suppose further that $\chi_0$ is triangulable.

\emph{Remark.} In focusing on the purifying manifold $\chi_0$, we are henceforth assuming that the bulk of $\chi$ (or $\mathfrak{X}$) consists of points with local Euclidean dimension $D$. Thus, in the measure-theoretic arguments to come, we will be neglecting any region of lower Euclidean dimension or points without a Euclidean neighborhood. We believe that this assumption is quite mild. Indeed, $j\circ i (P(\mathbb{C}^{d_A^2 d_B^2}))$ is diffeomorphic to the projective space, a smooth manifold. We assume that $\text{Tr}_{\mathscr{H}_A \otimes \mathscr{H}_B}^{-1}(\mathscr{K})$ has sufficient smoothness for the bulk of $\mathfrak{X}$ to consist of nonsingular points with a well-defined local Euclidean dimension, even if it changes from point to point. As for the lower dimensional Euclidean pieces, should they be (individually) immersed in $\chi_0$, they would have measure zero there. Hence, our conclusions will be invalidated for a given $\mathscr{K}$ only if lower dimensional pieces and/or non-Euclidean pieces of $\mathfrak{X}$ are somehow too numerous to be negligeable in a probabilistic argument. Finally, the assumption that $\chi_0$ is triangulable is made to simplify the definition of the measure $\tilde\mu$ on $\chi_0$. (See below.) Otherwise, the definition of $\tilde\mu$ must be modified to take the overlaps of a finite open cover into account. Note that the class of triangulable manifolds properly contains familiar classes of manifolds, namely differentiable, piecewise differentiable, and piecewise linear manifolds~\cite{Manolescu2014triangulations, Boothby1975introduction}. 

Because we assume the compact manifold $\chi_0$ to be triangulable, it can be covered by a finite collection of domains of integration $\{\mathcal{B}_i\}_{i=1,...,M}$, where each $\mathcal{B}_i$ is diffeomorphic to a closed $D$-dimensional ball $B^D (R)$ of volume $V_i \leq 1$, with $\sum_i V_i = 1$, and all $\mathcal{B}_i\cap\mathcal{B}_j$, $i\neq j$, are nullsets. By construction, $\chi_0$ is locally a (linear) purification space of dimension $D$ for the set $\mathscr{K}\subseteq\mathscr{D}$. Proposition~\ref{T:NE_lipshitz} applies with the relaxed Lipschitz constant $\eta = 4\sqrt{2}\log D$: for any $\mathcal{B}_i$,
\beq\label{E:Lipschitz_K}
|\NE\sigma - \NE\rho|\leq \eta \big\lVert \Ket{\pu{\sigma}}- \Ket{\pu{\rho}}\big\rVert_2^{(D)},
\eeq
with $D$-dimensional Euclidean norm $\lVert \cdot\rVert_2^{(D)}$, and where $\Ket{\pu{\sigma}}, \Ket{\pu{\rho}}$ are any purifications of $\sigma, \rho$ in $\mathcal{B}_i$. (Note that this choice of $\eta$ might be suboptimal as $D$ is the dimension of the \emph{purification} space.) Restricting the Euclidean measure of $B^D(R)$ to a sphere $S^{D-1}(r)$ produces a uniform measure on the sphere, i.e. a Haar measure. Conversely, the Euclidean measure on $B^D(R)$ can be recovered by integrating over spherical shells or radius $r$ and thickness $dr$ equipped with the measure $\text{vol}(S^{D-1}(r))\mu_{\text{Haar}} dr$, where $\text{vol}(\cdot)$ indicates Euclidean volume. These measures can be pulled back to the purifying manifold $\chi_0$. Suppose that the spherical shells of $\mathcal{B}_i$ are indexed by $k$ (whose range could be finite or infinite), and denote them $\mathcal{S}_{ik}$ with (pulled-back) Haar measure $\nu_{ik}$, volume $v_{ik}$ and mean NE $m_{ik}$. A normalized measure on $\chi_0$ can be defined as $\tilde\mu (E) = \sum_{ik}v_{ik} \tilde\mu_{ik}(E)$, where $\tilde\mu_{ik}(E) = \nu_{ik}(E\cap\mathcal{S}_{ik})$, and may be pushed forward to the image of $\chi_0$ in $\mathscr{K}$:
\beq\label{E:maps_2}
(\chi_0 = \cup_{ik}\mathcal{S}_{ik}\; , \tilde\mu ) \xrightarrow{\tau= \text{Tr}_{\mathscr{H}_A \otimes \mathscr{H}_B} \circ j}(\mathscr{K} , \mu).
\eeq
(Actually, $\mu$ may be extended to $\mathscr{K}$ by stating that any set in the complement of $\tau(\chi_0)$ has zero content.) We deduce from~\eqref{E:Lipschitz_K} that the NE will exhibit concentration around its mean $m_{ik}$ on each spherical shell of $\chi_0$:
\beq
\tilde\mu(\{\Ket{\Psi} \in \mathcal{S}_{ik} \colon | \NE{\tau(\Ket{\Psi})} - m_{ik} | > \alpha\}) \leq 2v_{ik} e^{-cD \alpha^2 / \eta^2}.
\eeq
The constant $c$ may be chosen as in Lemma~\ref{T:levy}. We conclude that the same will hold in $\mathscr{K}$:
\begin{proposition}\label{T:conc_NE_subset}
Let $\mathscr{K}\subseteq\mathscr{D}$ be a closed convex set of states with purification dimension $D$, and purifying manifold $\chi_0$. Let $\{\mathcal{S}_{ik}\}$ be a covering of $\chi_0$ by almost-nowhere-intersecting spherical shells. Then
\beq\label{E:conc_NE_subset}
\mu(\{\rho \in \tau(\mathcal{S}_{ik}) \colon | \NE\rho - m_{ik} | > \alpha\}) \leq 2v_{ik} e^{-cD \alpha^2 / \eta^2},
\eeq
where $\mu$ is the pushforward of $\tilde\mu$, $\tau$ is defined by~\eqref{E:maps_1},~\eqref{E:maps_2}, the mean NE is $m_{ik}=v_{ik}^{-1}\int_{\tau(\mathcal{S}_{ik})}d\mu(\sigma)\NE\sigma$, the volume of the shell is $v_{ik}=\int_{\tau(\mathcal{S}_{ik})}d\mu(\sigma)$, and $\eta = 4\sqrt{2}\log D$. The positive constant $c$ may be chosen as $c=(18\pi^3)^{-1}$. 

\end{proposition}

\subsection{Concentration of symmetry-induced entanglement in $\sepn$} 
In order to apply the concentration of NE, Proposition~\ref{T:conc_NE_subset}, to the problem of symmetric separability, we restrict our attention to the set $\sepn$. Because $\sepn$ is a closed convex set\footnote{It is the intersection of two closed convex sets, $\dn$ and $\sep$.}, it has an (integer) purification dimension $\text{p-dim }\sepn=D$, and a purifying manifold $\chi_0$, as defined in the previous section. Moreover, $\chi_0$ may be covered by almost-nowhere-intersecting spherical shells $\{\mathcal{S}_{ik}\}$ with respective volume $v_{ik}$, and mean NE $m_{ik}$. Let us finally look at the set
\beq\label{E:Y_alpha_N}
\{\rho \in\tau(\mathcal{S}_{ik}) \colon |\NE\rho - m_{ik} | > \alpha\}\subset\sepn,
\eeq
whose size is at most $2v_{ik}e^{-cD \alpha^2 / \eta^2}$, according to Proposition~\ref{T:conc_NE_subset}. 
\begin{proposition}\label{T:conc_NE_conditional}
Separable states that are symmetric for $\hat N$ present a strong concentration of NE around the mean $m_{ik}$ in each spherical shell of the purifying manifold: for $\rho\in \tau(\mathcal{S}_{ik})$,
\beq
\textup{Prob}_{\mu}(| \NE\rho - m_{ik} | > \alpha) \leq 2v_{ik}e^{-c D \alpha^2/\eta^2},
\eeq
where $m_{ik}=v_{ik}^{-1}\int_{\tau(\mathcal{S}_{ik})}d\mu(\sigma)\NE\sigma$, $D$ is the dimension of the purifying manifold, and $\eta = 4\sqrt{2}\log D$.
\end{proposition}
\begin{proposition}\label{T:nonzero_mean}
If there exists a state $\rho\in \tau(\mathcal{S}_{ik})$ for which $\NE\rho >0$, then the mean NE is nonzero on $\tau(\mathcal{S}_{ik})$. In other words, the mean NE is zero on $\tau(\mathcal{S}_{ik})$ iff $\NE\rho \equiv 0$ on $\tau(\mathcal{S}_{ik})$.
\end{proposition}
\begin{proof}
Suppose $\NE\rho = \epsilon >0$, and let $\sigma$ be $(\epsilon /\eta)$-close to $\rho$ in Bures distance, $D_B(\sigma , \rho) < \epsilon / \eta$. From Corollary~\ref{T:Bures}, $|\NE\sigma - \NE\rho|<\epsilon$ so $\NE\sigma >0$. Therefore the set $S=\{\rho\in\tau(\mathcal{S}_{ik}) \colon \NE\rho >0\}$ is open. If nonempty, we have $\mu(S)>0$, and the mean NE is
\beq
m_{ik} = \int_{\tau(\mathcal{S}_{ik})}d\mu(\sigma)\NE\sigma  = \int_S d\mu(\sigma)\NE\sigma  >0.
\eeq
\end{proof}
Because $\symsepn$ is of zero measure in $\sepn$ (unless $\hat N$ is nondegenerate), Proposition~\ref{T:nonzero_mean} indicates that a generic $\hat N$ will have a nonzero mean NE, $m_{ik}>0$, on each $\tau(\mathcal{S}_{ik})$. We contend that these local means strongly concentrate around $m=\sum_{ik} v_{ik}m_{ik}$, the global mean NE on $\sepn$. Indeed, by Proposition~\ref{T:conc_NE_conditional} the NE strongly concentrates around its local mean $m_{ik}$ on each shell.\footnote{We also contend that the NE will have asymptotically vanishing variance on each shell. Indeed, using the cumulative distribution function $v_{ik}^{-1}\textup{Prob}_{\mu}(| \NE\rho - m_{ik} | \leq \alpha)$, we may express the variance of NE on shell $\tau(\mathcal{S}_{ik})$ as:
\beq
\begin{aligned}
\text{Var}(\NE\rho) &= \mathbb{E}\left((\NE\rho - m_{ik})^2\right)\\
&=v_{ik}^{-1}\int_0^{\infty}d\alpha\;\alpha^2\frac{d}{d\alpha}\textup{Prob}_{\mu}(| \NE\rho - m_{ik} | \leq \alpha).
\end{aligned}
\eeq   
The cumulative distribution function is monotone increasing and, following Proposition~\ref{T:conc_NE_conditional}, it is comprised in the tight interval between $\ell(\alpha)=1-2e^{-c D \alpha^2 /\eta^2}$ and 1. Thus its slope can be large only on a small set, and on average is expected to be close to that of $\ell(\alpha)$. Therefore,
\beq
\text{Var}(\NE\rho) \approx \int_0^{\infty}d\alpha\;\alpha^2\frac{d}{d\alpha}\ell(\alpha)=O\left((\log D)^2 /D\right).
\eeq
} By (Lipschitz) continuity, Proposition~\ref{T:NE_lipshitz}, the NE must agree where shells meet, and this in turn forces local means $m_{ik}$ to remain close to the global mean $m=\sum_{ik} v_{ik}m_{ik}$. The phenomenon also tends to keep their variance $\text{Var}(\{m_{ik}\})$ small.

Define $\epsilon_{ik}=|m-m_{ik}|$. It follows from what has just been said that most of the $\epsilon_{ik}$ are small, and from the properties of the variance, that $\text{Var}(\{\epsilon_{ik}\})=\text{Var}(\{m_{ik}\})$ is also small. Now for any $\alpha > \epsilon_{ik}$, we have $|\NE\rho - m|>\alpha \implies |\NE\rho - m_{ik}|>\alpha - \epsilon_{ik}$, so on each shell $\tau(\mathcal{S}_{ik})$ we get from Proposition~\ref{T:conc_NE_conditional}:
\beq
\textup{Prob}_{\mu}(| \NE\rho - m | > \alpha) \leq \textup{Prob}_{\mu}(| \NE\rho - m_{ik} | > \alpha - \epsilon_{ik}) \leq 2v_{ik}e^{-c D (\alpha - \epsilon_{ik})^2/\eta^2}.
\eeq
Summing up the contributions from all shells, we find
\beq\label{E:prob_all}
\textup{Prob}_{\mu}(| \NE\rho - m | > \alpha) \leq 2\sum_{ik} v_{ik} e^{-c D (\alpha - \epsilon_{ik})^2 /\eta^2},
\eeq
for a $\mu$-random states $\rho\in\sepn$. The sum may be upper bounded by a generalized Jensen's inequality~\cite{Liao2017sharpening},
\beq
\mathbb{E}h(X)\leq h(\mathbb{E}X)+\text{Var}(X)\,\text{sup}\frac{h''(x)}{2},
\eeq
with $h(x)=\exp (-cDx^2/\eta^2)$, and random variable $X$ having value $\alpha - \epsilon_{ik}$ with probability $v_{ik}$. We find
\beq\label{E:var_appendix}
\sum_{ik} v_{ik}e^{-c D (\alpha - \epsilon_{ik})^2 /\eta^2}
\leq e^{-c D(\alpha - \bar\epsilon)^2/\eta^2} + \text{Var}(\{\epsilon_{ik}\})O\left(\frac{D^3}{\eta^6} e^{-O(D^2/\eta^4)}\right),
\eeq
where $\bar\epsilon=\sum_{ik} v_{ik}\epsilon_{ik} \to 0$. Remembering that $\eta = O(\log D)$, we see that the last term vanishes even for moderate values of $D$, unless the $\epsilon_{ik}$'s have a violently diverging variance. We have argued above that such was not the case. We conclude that~\eqref{E:prob_all} can be restated as:
\begin{proposition}\label{T:genN}
For an observable $\hat N$ with large enough purification dimension $D$, the NE on $\sepn$ strongly concentrates around its mean value:
\beq
\textup{Prob}_{\mu}(| \NE\rho - m | > \alpha) \lesssim O(e^{-c D \alpha^2 /\eta^2}),
\eeq
where $m$ is the (strictly positive) mean NE of $\hat N$ averaged over all states in $\sepn$, $c$ is a positive constant that may be chosen as $c=(18\pi^3)^{-1}$, and $\eta = 4\sqrt{2}\log D$.
\end{proposition}

Our analytical findings are confirmed by numerics, as depicted in Figs.~\ref{F:numerics_qudits_sum},~\ref{F:numerics_qudits_prod},~\ref{F:numerics_qubits_sum_prod}. Random separable states were generated (through convex combinations of Haar-random pure product states, and using Carathéodory's theorem on extremal points on convex sets), before they were submitted to nonselective $\hat N$-measurement to produce random states in $\sepn$. The distribution of NE values was then plotted for different system dimensions, showing a fast narrowing around the mean as dimension is increased. The distributions are empirically found to closely follow a chi distribution of noninteger order $k$ depending on dimension in a way that we cannot yet explain analytically. Let us simply recall that the chi distribution of integer order $k$ corresponds to the Euclidean distance between the origin and a $k$-tuple of normal random variables. The connection to our random states might stem from their being build out of Haar-random pure states whose normally distributed components are submitted to projection, and nonlinear operations yielding von Neumann entropies. If the putative connection is valid, the noninteger order $k$ is bound to depend on (the degeneracy sectors of) $\hat N$. 

\begin{figure}
\includegraphics[width=\textwidth]{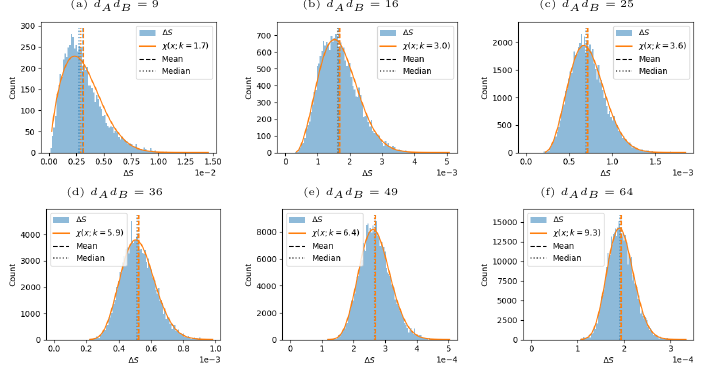}
\caption{Distribution of NE values for 2-qudit states. Each subsystem $A$, $B$ supports a single qudit of dimension $d_A = d_B =d \in \{2,3,4\}$, and $\hat N_A \otimes 1_B + 1_A \otimes \hat N_B$ with $\hat N_A$, $\hat N_B$ possessing nondegenerate eigenvalues $\{1,2,...,d\}$, e.g. single-qudit level-number on each subsystem. The distributions are fitted to an empirical chi distribution of order $k$ that depends on dimension.}
\label{F:numerics_qudits_prod}
\end{figure}

\begin{figure}
\includegraphics[width=\textwidth]{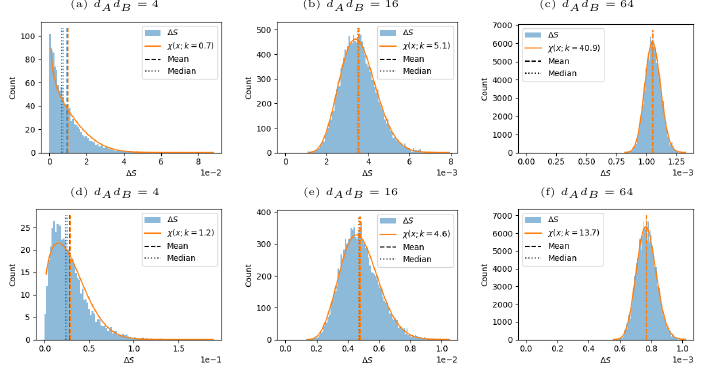}
\caption{Distribution of NE values for multi-qubit states. Each subsystem $A$, $B$ supports $q$ qubits for a total dimension $d_A d_B = 2^{2q}$, $q=1,2,3$. $\hat N_A$, $\hat N_B$ have the form $\sum \sigma_{z}$. \textbf{(a-c)} $\hat N = \hat N_A \otimes \hat N_B$. \textbf{(d-f)} $\hat N = \hat N_A \otimes 1_B + 1_A \otimes N_B$. The distributions are fitted to an empirical chi distribution of order $k$ that depends on dimension.}
\label{F:numerics_qubits_sum_prod}
\end{figure}

\bibliographystyle{apsrev4-1}
\bibliography{refsNE}

@article{Bartlett2007reference,
  title = {Reference frames, superselection rules, and quantum information},
  author = {Bartlett, Stephen D. and Rudolph, Terry and Spekkens, Robert W.},
  journal = {Rev. Mod. Phys.},
  volume = {79},
  issue = {2},
  pages = {555--609},
  numpages = {0},
  year = {2007},
  month = {Apr},
  publisher = {American Physical Society},
  doi = {10.1103/RevModPhys.79.555},
  url = {https://link.aps.org/doi/10.1103/RevModPhys.79.555}
}

@inproceedings{Bennett1984quantum,
author = {Bennett, Charles and Brassard, Gilles},
year = {1984},
month = {12},
pages = {175-179},
title = {Quantum cryptography: Public key distribution and coin tossing},
volume = {1},
booktitle = {Proceedings of IEEE International Conference on Computers, Systems and Signal Processing},
location = {Bangalore, India},
url = {https://web.archive.org/web/20200130165639/http://researcher.watson.ibm.com/researcher/files/us-bennetc/BB84highest.pdf}
}

@article{Bennett1993teleporting,
  title = {Teleporting an unknown quantum state via dual classical and Einstein-Podolsky-Rosen channels},
  author = {Bennett, Charles H. and Brassard, Gilles and Cr\'epeau, Claude and Jozsa, Richard and Peres, Asher and Wootters, William K.},
  journal = {Phys. Rev. Lett.},
  volume = {70},
  issue = {13},
  pages = {1895--1899},
  numpages = {0},
  year = {1993},
  month = {Mar},
  publisher = {American Physical Society},
  doi = {10.1103/PhysRevLett.70.1895},
  url = {https://link.aps.org/doi/10.1103/PhysRevLett.70.1895}
}

@article{Bhatia2019bures,
title = {On the Bures–Wasserstein distance between positive definite matrices},
journal = {Expositiones Mathematicae},
volume = {37},
number = {2},
pages = {165-191},
year = {2019},
issn = {0723-0869},
doi = {https://doi.org/10.1016/j.exmath.2018.01.002},
url = {https://www.sciencedirect.com/science/article/pii/S0723086918300021},
author = {Rajendra Bhatia and Tanvi Jain and Yongdo Lim}
}

@book{Boothby1975introduction,
  title={An introduction to differentiable manifolds and Riemannian geometry},
  author={W. M. Boothby},
  year={1975},
  publisher={Academic Press, Elsevier Science},
  url={https://api.semanticscholar.org/CorpusID:120912603}
}

@Article{Bouwmeester1997experimental,
  author={Dik Bouwmeester and Jian-Wei Pan and Klaus Mattle and Manfred Eibl and Harald Weinfurter and Anton Zeilinger},
  title={{Experimental quantum teleportation}},
  journal={Nature},
  year=1997,
  volume={390},
  number={6660},
  pages={575-579},
  month={December},
  doi={10.1038/37539},
  url={https://ideas.repec.org/a/nat/nature/v390y1997i6660d10.1038_37539.html}
}

@article{Braunstein1996geometry,
title = {Geometry of quantum inference},
journal = {Physics Letters A},
volume = {219},
number = {3},
pages = {169-174},
year = {1996},
issn = {0375-9601},
doi = {https://doi.org/10.1016/0375-9601(96)00365-9},
url = {https://www.sciencedirect.com/science/article/pii/0375960196003659},
author = {Samuel L. Braunstein}
}

@article{Bremner2009random,
  title = {Are Random Pure States Useful for Quantum Computation?},
  author = {Bremner, Michael J. and Mora, Caterina and Winter, Andreas},
  journal = {Phys. Rev. Lett.},
  volume = {102},
  issue = {19},
  pages = {190502},
  numpages = {4},
  year = {2009},
  month = {May},
  publisher = {American Physical Society},
  doi = {10.1103/PhysRevLett.102.190502},
  url = {https://link.aps.org/doi/10.1103/PhysRevLett.102.190502}
}

@article{Das2021universal,
  title = {Universal Limitations on Quantum Key Distribution over a Network},
  author = {Das, Siddhartha and B\"auml, Stefan and Winczewski, Marek and Horodecki, Karol},
  journal = {Phys. Rev. X},
  volume = {11},
  issue = {4},
  pages = {041016},
  numpages = {38},
  year = {2021},
  month = {Oct},
  publisher = {American Physical Society},
  doi = {10.1103/PhysRevX.11.041016},
  url = {https://link.aps.org/doi/10.1103/PhysRevX.11.041016}
}

@article{Gharibian2010strong,
author = {Gharibian, Sevag},
title = {Strong NP-hardness of the quantum separability problem},
year = {2010},
issue_date = {March 2010},
publisher = {Rinton Press, Incorporated},
address = {Paramus, NJ},
volume = {10},
number = {3},
issn = {1533-7146},
journal = {Quantum Info. Comput.},
month = mar,
pages = {343–360},
numpages = {18}
}

@article{Gross2009quantum,
  title = {Most Quantum States Are Too Entangled To Be Useful As Computational Resources},
  author = {Gross, D. and Flammia, S. T. and Eisert, J.},
  journal = {Phys. Rev. Lett.},
  volume = {102},
  issue = {19},
  pages = {190501},
  numpages = {4},
  year = {2009},
  month = {May},
  publisher = {American Physical Society},
  doi = {10.1103/PhysRevLett.102.190501},
  url = {https://link.aps.org/doi/10.1103/PhysRevLett.102.190501}
}

@inproceedings{Grover1996fast,
author = {Grover, Lov K.},
title = {A fast quantum mechanical algorithm for database search},
year = {1996},
isbn = {0897917855},
publisher = {Association for Computing Machinery},
address = {New York, NY, USA},
url = {https://doi.org/10.1145/237814.237866},
doi = {10.1145/237814.237866},
booktitle = {Proceedings of the Twenty-Eighth Annual ACM Symposium on Theory of Computing},
pages = {212–219},
numpages = {8},
location = {Philadelphia, Pennsylvania, USA},
series = {STOC '96}
}

@article{Gurvits2002largest,
  title = {Largest separable balls around the maximally mixed bipartite quantum state},
  author = {Gurvits, Leonid and Barnum, Howard},
  journal = {Phys. Rev. A},
  volume = {66},
  issue = {6},
  pages = {062311},
  numpages = {7},
  year = {2002},
  month = {Dec},
  publisher = {American Physical Society},
  doi = {10.1103/PhysRevA.66.062311},
  url = {https://link.aps.org/doi/10.1103/PhysRevA.66.062311}
}

@inproceedings{Gurvits2003classical,
author = {Gurvits, Leonid},
title = {Classical deterministic complexity of Edmonds' Problem and quantum entanglement},
year = {2003},
isbn = {1581136749},
publisher = {Association for Computing Machinery},
address = {New York, NY, USA},
url = {https://doi.org/10.1145/780542.780545},
doi = {10.1145/780542.780545},
booktitle = {Proceedings of the Thirty-Fifth Annual ACM Symposium on Theory of Computing},
pages = {10–19},
numpages = {10},
location = {San Diego, CA, USA},
series = {STOC '03}
}

@article{Hall1998random,
title = {Random quantum correlations and density operator distributions},
journal = {Physics Letters A},
volume = {242},
number = {3},
pages = {123-129},
year = {1998},
issn = {0375-9601},
doi = {https://doi.org/10.1016/S0375-9601(98)00190-X},
url = {https://www.sciencedirect.com/science/article/pii/S037596019800190X},
author = {Michael J.W. Hall}
}

@article{Hayden2006aspects,
   title={Aspects of Generic Entanglement},
   volume={265},
   ISSN={1432-0916},
   url={http://dx.doi.org/10.1007/s00220-006-1535-6},
   DOI={10.1007/s00220-006-1535-6},
   number={1},
   journal={Communications in Mathematical Physics},
   publisher={Springer Science and Business Media LLC},
   author={Hayden, Patrick and Leung, Debbie W. and Winter, Andreas},
   year={2006},
   month=mar, pages={95–117} }

@book{Hayashi2006quantum,
author = {Hayashi, Masahito},
year = {2006},
month = {01},
pages = {},
title = {Quantum Information: An Introduction},
publisher = {Springer Berlin, Heidelberg},
isbn = {978-3-540-30265-0},
journal = {Quantum Information: An Introduction},
doi = {10.1007/3-540-30266-2}
}

@article{Howard2014contextuality,
author = {Howard, Mark and Wallman, Joel and Veitch, Victor and Emerson, Joseph},
year = {2014},
month = {01},
pages = {},
title = {Contextuality supplies the magic for quantum computation},
volume = {510},
journal = {Nature},
doi = {10.1038/nature13460}
}

@article{Hyllus2012Fisher,
  title = {Fisher information and multiparticle entanglement},
  author = {Hyllus, Philipp and Laskowski, Wies\l{}aw and Krischek, Roland and Schwemmer, Christian and Wieczorek, Witlef and Weinfurter, Harald and Pezz\'e, Luca and Smerzi, Augusto},
  journal = {Phys. Rev. A},
  volume = {85},
  issue = {2},
  pages = {022321},
  numpages = {10},
  year = {2012},
  month = {Feb},
  publisher = {American Physical Society},
  doi = {10.1103/PhysRevA.85.022321},
  url = {https://link.aps.org/doi/10.1103/PhysRevA.85.022321}
}

@book{Ledoux2001concentration,
author = {Ledoux, M.},
year = {2001},
month = {01},
pages = {},
title = {The Concentration of Measure Phenomenon},
publisher = {American Mathematical Society}
}

@article{Liao2017sharpening,
  title={Sharpening Jensen's Inequality},
  author={Jason Liao and Arthur Berg},
  journal={The American Statistician},
  year={2017},
  volume={73},
  pages={278 - 281},
  url={https://api.semanticscholar.org/CorpusID:88515366}
}

@article{Lloyd1996universal,
  title={Universal Quantum Simulators},
  author={Seth Lloyd},
  journal={Science},
  year={1996},
  volume={273},
  pages={1073 - 1078},
  url={https://api.semanticscholar.org/CorpusID:43496899}
}

@article{Lockhart2002preserving,
  title = {Preserving entanglement under perturbation and sandwiching all separable states},
  author = {Lockhart, Robert B. and Steiner, Michael J.},
  journal = {Phys. Rev. A},
  volume = {65},
  issue = {2},
  pages = {022107},
  numpages = {4},
  year = {2002},
  month = {Jan},
  publisher = {American Physical Society},
  doi = {10.1103/PhysRevA.65.022107},
  url = {https://link.aps.org/doi/10.1103/PhysRevA.65.022107}
}

@article{Ma2022symmetric,
  title = {Symmetric inseparability and number entanglement in charge-conserving mixed states},
  author = {Ma, Zhanyu and Han, Cheolhee and Meir, Yigal and Sela, Eran},
  journal = {Phys. Rev. A},
  volume = {105},
  issue = {4},
  pages = {042416},
  numpages = {12},
  year = {2022},
  month = {Apr},
  publisher = {American Physical Society},
  url = {https://journals.aps.org/pra/abstract/10.1103/PhysRevA.105.042416}
}

@article{Macieszczak2019coherence,
author = {Macieszczak, Katarzyna and Levi, Emanuele and Macri, Tommaso and Lesanovsky, Igor and Garrahan, Juan},
year = {2019},
month = {05},
pages = {},
title = {Coherence, entanglement, and quantumness in closed and open systems with conserved charge, with an application to many-body localization},
volume = {99},
journal = {Physical Review A},
doi = {10.1103/PhysRevA.99.052354}
}

@inproceedings{Manolescu2014triangulations,
  title={Triangulations of Manifolds},
  author={Ciprian Manolescu},
  year={2014},
  booktitle={Notices of the International Congress of Chinese Mathematicians},
  volume={2},
  pages={21-23},
  url={https://link.intlpress.com/JDetail/1806612934258262017}
}

@book{Nielsen_Chuang2010quantum, 
place={Cambridge}, 
title={Quantum Computation and Quantum Information: 10th Anniversary Edition}, 
publisher={Cambridge University Press}, 
author={Nielsen, Michael A. and Chuang, Isaac L.}, 
year={2010}
}

@article{Palazuelos2022genuine,
  doi = {10.22331/q-2022-06-13-735},
  url = {https://doi.org/10.22331/q-2022-06-13-735},
  title = {Genuine multipartite entanglement of quantum states in the multiple-copy scenario},
  author = {Palazuelos, Carlos and Vicente, Julio I. de},
  journal = {{Quantum}},
  issn = {2521-327X},
  publisher = {{Verein zur F{\"{o}}rderung des Open Access Publizierens in den Quantenwissenschaften}},
  volume = {6},
  pages = {735},
  month = jun,
  year = {2022}
}

@inproceedings{Parez2024fate,
  title={The Fate of Entanglement},
  author={Parez, Gilles and Witczak-Krempa, William},
  year={2024},
  url={https://api.semanticscholar.org/CorpusID:267627292}
}

@article{Popescu1994Bell,
  title = {Bell's inequalities versus teleportation: What is nonlocality?},
  author = {Popescu, Sandu},
  journal = {Phys. Rev. Lett.},
  volume = {72},
  issue = {6},
  pages = {797--799},
  numpages = {0},
  year = {1994},
  month = {Feb},
  publisher = {American Physical Society},
  doi = {10.1103/PhysRevLett.72.797},
  url = {https://link.aps.org/doi/10.1103/PhysRevLett.72.797}
}

@article{Popescu1995bell,
  title = {Bell's Inequalities and Density Matrices: Revealing ``Hidden'' Nonlocality},
  author = {Popescu, Sandu},
  journal = {Phys. Rev. Lett.},
  volume = {74},
  issue = {14},
  pages = {2619--2622},
  numpages = {0},
  year = {1995},
  month = {Apr},
  publisher = {American Physical Society},
  doi = {10.1103/PhysRevLett.74.2619},
  url = {https://link.aps.org/doi/10.1103/PhysRevLett.74.2619}
}

@article{Raussendorf2013contextuality,
  title = {Contextuality in measurement-based quantum computation},
  author = {Raussendorf, Robert},
  journal = {Phys. Rev. A},
  volume = {88},
  issue = {2},
  pages = {022322},
  numpages = {7},
  year = {2013},
  month = {Aug},
  publisher = {American Physical Society},
  doi = {10.1103/PhysRevA.88.022322},
  url = {https://link.aps.org/doi/10.1103/PhysRevA.88.022322}
}

@article{Schuh2004nonlocal,
  title = {Nonlocal Resources in the Presence of Superselection Rules},
  author = {Schuch, N. and Verstraete, F. and Cirac, J. I.},
  journal = {Phys. Rev. Lett.},
  volume = {92},
  issue = {8},
  pages = {087904},
  numpages = {4},
  year = {2004},
  month = {Feb},
  publisher = {American Physical Society},
  doi = {10.1103/PhysRevLett.92.087904},
  url = {https://link.aps.org/doi/10.1103/PhysRevLett.92.087904}
}

@article{Schuh2004quantum,
  title = {Quantum entanglement theory in the presence of superselection rules},
  author = {Schuch, Norbert and Verstraete, Frank and Cirac, J. Ignacio},
  journal = {Phys. Rev. A},
  volume = {70},
  issue = {4},
  pages = {042310},
  numpages = {15},
  year = {2004},
  month = {Oct},
  publisher = {American Physical Society},
  doi = {10.1103/PhysRevA.70.042310},
  url = {https://link.aps.org/doi/10.1103/PhysRevA.70.042310}
}

@article{Seevinck2001sufficent,
  title = {Sufficient conditions for three-particle entanglement and their tests in recent experiments},
  author = {Seevinck, Michael and Uffink, Jos},
  journal = {Phys. Rev. A},
  volume = {65},
  issue = {1},
  pages = {012107},
  numpages = {7},
  year = {2001},
  month = {Dec},
  publisher = {American Physical Society},
  doi = {10.1103/PhysRevA.65.012107},
  url = {https://link.aps.org/doi/10.1103/PhysRevA.65.012107}
}

@article{Shapourian2021entanglement,
  title = {Entanglement Negativity Spectrum of Random Mixed States: A Diagrammatic Approach},
  author = {Shapourian, Hassan and Liu, Shang and Kudler-Flam, Jonah and Vishwanath, Ashvin},
  journal = {PRX Quantum},
  volume = {2},
  issue = {3},
  pages = {030347},
  numpages = {28},
  year = {2021},
  month = {Sep},
  publisher = {American Physical Society},
  doi = {10.1103/PRXQuantum.2.030347},
  url = {https://link.aps.org/doi/10.1103/PRXQuantum.2.030347}
}

@INPROCEEDINGS{Shor1994algorithms,
  author={Shor, P.W.},
  booktitle={Proceedings 35th Annual Symposium on Foundations of Computer Science}, 
  title={Algorithms for quantum computation: discrete logarithms and factoring}, 
  year={1994},
  volume={},
  number={},
  pages={124-134},
  doi={10.1109/SFCS.1994.365700}}

@article{Steiner2003generalized,
  title = {Generalized robustness of entanglement},
  author = {Steiner, Michael},
  journal = {Phys. Rev. A},
  volume = {67},
  issue = {5},
  pages = {054305},
  numpages = {4},
  year = {2003},
  month = {May},
  publisher = {American Physical Society},
  doi = {10.1103/PhysRevA.67.054305},
  url = {https://link.aps.org/doi/10.1103/PhysRevA.67.054305}
}

@article{Szarek2005volume,
  title = {Volume of separable states is super-doubly-exponentially small in the number of qubits},
  author = {Szarek, Stanislaw J.},
  journal = {Phys. Rev. A},
  volume = {72},
  issue = {3},
  pages = {032304},
  numpages = {10},
  year = {2005},
  month = {Sep},
  publisher = {American Physical Society},
  doi = {10.1103/PhysRevA.72.032304},
  url = {https://link.aps.org/doi/10.1103/PhysRevA.72.032304}
}

@article{Verstraete2003quantum,
  title = {Quantum Nonlocality in the Presence of Superselection Rules and Data Hiding Protocols},
  author = {Verstraete, F. and Cirac, J. I.},
  journal = {Phys. Rev. Lett.},
  volume = {91},
  issue = {1},
  pages = {010404},
  numpages = {4},
  year = {2003},
  month = {Jul},
  publisher = {American Physical Society},
  doi = {10.1103/PhysRevLett.91.010404},
  url = {https://link.aps.org/doi/10.1103/PhysRevLett.91.010404}
}

@article{Vidal1999robustness,
  title = {Robustness of entanglement},
  author = {Vidal, Guifr\'e and Tarrach, Rolf},
  journal = {Phys. Rev. A},
  volume = {59},
  issue = {1},
  pages = {141--155},
  numpages = {0},
  year = {1999},
  month = {Jan},
  publisher = {American Physical Society},
  doi = {10.1103/PhysRevA.59.141},
  url = {https://link.aps.org/doi/10.1103/PhysRevA.59.141}
}

@article{Wen2023separable,
doi = {10.1088/1751-8121/ace810},
url = {https://dx.doi.org/10.1088/1751-8121/ace810},
year = {2023},
month = {jul},
publisher = {IOP Publishing},
volume = {56},
number = {33},
pages = {335302},
author = {Wen, Robin Yunfei and Kempf, Achim},
title = {Separable ball around any full-rank multipartite product state},
journal = {Journal of Physics A: Mathematical and Theoretical}
}

@book{Zeng2015quantum,
  title={Quantum Information Meets Quantum Matter},
  author={Bei Zeng and Xie Chen and Duanlu Zhou and Xiao-Gang Wen},
  journal={Quantum Science and Technology},
  publisher = {Springer New York, NY},
  year={2015},
  url={https://api.semanticscholar.org/CorpusID:118528258}
}

@article{Zyczkowski1998volume1,
    author = "Zyczkowski, Karol and Horodecki, Pawel and Sanpera, Anna and Lewenstein, Maciej",
    title = "{On the volume of the set of mixed entangled states}",
    eprint = "quant-ph/9804024",
    archivePrefix = "arXiv",
    doi = "10.1103/PhysRevA.58.883",
    journal = "Phys. Rev. A",
    volume = "58",
    pages = "883",
    year = "1998"
}

@article{Zyczkowski1999volume2,
    author = "Zyczkowski, Karol",
    title = "{On the volume of the set of mixed entangled states. 2.}",
    eprint = "quant-ph/9902050",
    archivePrefix = "arXiv",
    doi = "10.1103/PhysRevA.60.3496",
    journal = "Phys. Rev. A",
    volume = "60",
    pages = "3496",
    year = "1999"
}

@article{Zyczkowski2001induced,
doi = {10.1088/0305-4470/34/35/335},
url = {https://dx.doi.org/10.1088/0305-4470/34/35/335},
year = {2001},
month = {aug},
publisher = {},
volume = {34},
number = {35},
pages = {7111},
author = {Karol Zyczkowski and  Hans-Jürgen Sommers},
title = {Induced measures in the space of mixed quantum states},
journal = {Journal of Physics A: Mathematical and General}
}

\end{document}